\documentclass[preprint]{elsarticle}
\usepackage[top=2cm,bottom=3cm,left=3cm,right=2cm]{geometry}
\usepackage[utf8]{inputenc}
\usepackage{amsmath}
\usepackage{amssymb}
\usepackage{amsthm}
\usepackage{mathtools}
\usepackage{slashbox}
\usepackage{graphicx}
\usepackage{tabularx}
\date{}
\newtheorem{theorem}{Theorem}
\newtheorem{conjecture}{Conjecture}
\newtheorem{lemma}{Lemma}
\newtheorem{remark}{Remark}
\newtheorem{corollary}{Corollary}

\begin{document}
\begin{frontmatter}

\title{Interval edge-colorings of complete graphs}

\author[ysu]{H. H. Khachatrian}
\ead{hrant.khachatrian@ysu.am}
\author[ysu,sci]{P. A. Petrosyan}
\ead{pet\_petros@ipia.sci.am}
\address[ysu]{Department of Informatics and Applied Mathematics, Yerevan State University, Yerevan, 0025, Armenia}
\address[sci]{Institute for Informatics and Automation Problems,
National Academy of Sciences, Yerevan, 0014, Armenia}

\begin{abstract}
An edge-coloring of a graph $G$ with colors $1,2,\ldots,t$ is an \textit{interval $t$-coloring} if all colors are used, and the colors of edges incident to each vertex of $G$ are distinct and form an interval of integers. A graph $G$ is \textit{interval colorable} if it has an interval $t$-coloring for some positive integer $t$. For an interval colorable graph $G$, $W(G)$ denotes the greatest value of $t$ for which $G$ has an interval $t$-coloring. It is known that the complete graph is interval colorable if and only if the number of its vertices is even. However, the exact value of $W(K_{2n})$ is known only for $n \leq 4$. The second author showed that if $n = p2^q$, where $p$ is odd and $q$ is nonnegative, then $W(K_{2n}) \geq 4n-2-p-q$. Later, he conjectured that if $n \in \mathbb{N}$, then $W(K_{2n}) = 4n - 2 - \left\lfloor\log_2{n}\right\rfloor - \left \| n_2 \right \|$, where $\left \| n_2 \right \|$ is the number of $1$'s in the binary representation of $n$.

In this paper we introduce a new technique to construct interval colorings of complete graphs based on their 1-factorizations, which is used to disprove the conjecture, improve lower and upper bounds on $W(K_{2n})$ and determine its exact values for $n \leq 12$.
\end{abstract}
\end{frontmatter}

\section{Introduction}

All graphs in this paper are finite, undirected, have no loops or multiple edges. Let $V(G)$ and $E(G)$ denote the sets of vertices and edges of a graph $G$, respectively. For $S\subseteq V(G)$, $G[S]$ denotes the subgraph of $G$ induced by $S$, that is, $V(G[S])=S$ and $E(G[S])$ consists of those edges of $E(G)$ for which both ends are
in $S$.
For a graph $G$, $\Delta(G)$ denotes the maximum degree of vertices in $G$. A graph $G$ is \textit{$r$-regular} if all its vertices have degree $r$. The set of edges $M$ is called a \textit{matching} if no two edges from $M$ are adjacent. A vertex $v$ is \textit{covered} by the matching $M$ if it is incident to one of the edges of $M$. 
A matching $M$ is a \textit{perfect matching} if it covers all the vertices of the graph $G$. The set of perfect matchings $\mathfrak{F} = \{F_1,F_2,\ldots,F_n\}$ is a \textit{1-factorization} of $G$ if every edge of $G$ belongs to exactly one of the perfect matchings in $\mathfrak{F}$. The set of integers $\{a,a+1,\ldots,b\}$, $a \leq b$, is denoted by $[a,b]$. The terms, notations and concepts that we do not define can be found in \cite{West}.

A \textit{proper edge-coloring} of graph $G$ is a coloring of the edges of $G$ such that no two adjacent edges receive the same color. The \textit{chromatic index} $\chi'(G)$ of a graph $G$ is the minimum number of colors used in a proper edge-coloring of $G$. If $\alpha$ is a proper edge-coloring of $G$ and $v \in V(G)$, then the \textit{spectrum} of a vertex $v$, denoted by $S(v,\alpha)$, is the set of colors of edges incident to $v$. By $\underline{S}(v,\alpha)$ and $\overline{S}(v,\alpha)$ we denote the the smallest and largest colors of the spectrum, respectively. If $\alpha$ is a proper edge-coloring of $G$ and $H$ is a subgraph of $G$, then we can define a union and intersection of spectrums of the vertices of $H$:
\begin{center}
$S_\cap(H, \alpha) = \bigcap\limits_{v \in V(H)} S(v, \alpha)$ \\
$S_\cup(H, \alpha) = \bigcup\limits_{v \in V(H)} S(v, \alpha)$ \\
\end{center}

A proper edge-coloring of a graph $G$ with colors $1,2,\ldots,t$ is an \textit{interval $t$-coloring} if all colors are used, and for any vertex $v$ of $G$, the set $S(v, \alpha)$ is an interval of consecutive integers. A graph $G$ is \textit{interval colorable} if it has an interval $t$-coloring for some positive integer $t$. The set of interval colorable graphs is denoted by $\mathfrak{N}$. For a graph $G \in \mathfrak{N}$, the least and the greatest values of $t$ for which $G$ has an interval $t$-coloring are denoted by $w(G)$ and $W(G)$, respectively. 

The concept of interval edge-coloring was introduced by Asratian and Kamalian \cite{AsratianKamalian1987}. In \cite{AsratianKamalian1987,AsratianKamalian1994}, they proved that if $G$ is interval colorable, then $\chi'(G)=\Delta(G)$. For regular graphs the converse is also true. Moreover, if $G \in \mathfrak{N}$ is regular, then $w(G) = \Delta(G)$ and $G$ has an interval $t$-coloring for every $t$, $w(G) \leq t \leq W(G)$. For a complete graph $K_m$, Vizing \cite{Vizing1965} proved that $\chi'(K_m)=m-1$ if $m$ is even and $\chi'(K_m)=m$ if $m$ is odd. These results imply that the complete graph is interval colorable if and only if the number of vertices is even. Moreover, $w(K_{2n}) = 2n-1$, for any $n \in \mathbb{N}$. On the other hand, the problem of determining the exact value of $W(K_{2n})$ is open since 1990.

In \cite{Kamalian1990} Kamalian proved the following upper bound on $W(G)$:
\begin{theorem}
If $G$ is a connected graph with at least two vertices and $G\in \mathfrak{N}$, then $W(G) \leq 2|V(G)|-3$.
\label{upperKamalian}
\end{theorem}

This upper bound was improved by Giaro, Kubale, Malafiejski in \cite{GiaroKubMal2001}:

\begin{theorem}\label{tGiaroKubale4n4}
If $G$ is a connected graph with at least three vertices and $G\in \mathfrak{N}$, then $W(G) \leq 2|V(G)|-4$.
\end{theorem}

Improved upper bounds on $W(G)$ are known for several classes of graphs, including triangle-free graphs \cite{AsratianKamalian1987,AsratianKamalian1994}, planar graphs \cite{Axenovich2002} and $r$-regular graphs with at least $2r+2$ vertices \cite{KamalianPetrosyan2012}. The exact value of the parameter $W(G)$ is known for even cycles, trees \cite{Kamalian1989}, complete bipartite graphs \cite{Kamalian1989}, Möbius ladders \cite{Petrosyan2005} and $n$-dimensional cubes \cite{Petrosyan2010, PetrosyanKhachatrianTananyan2013}. This paper is focused on investigation of $W(K_{2n})$.

The first lower bound on $W(K_{2n})$ was obtained by Kamalian in \cite{Kamalian1990}:
\begin{theorem}
For any $n\in \mathbb{N}$, $W(K_{2n}) \geq 2n-1 + \lfloor \log_2{(2n-1)} \rfloor$.
\end{theorem}

This bound was improved by the second author in \cite{Petrosyan2010}:
\begin{theorem}\label{tPetrosyan3n2}
For any $n\in \mathbb{N}$, $W(K_{2n}) \geq 3n-2$.
\end{theorem}

In the same paper he also proved the following statement:
\begin{theorem}\label{tPetrosyan4n}
For any $n\in \mathbb{N}$, $W(K_{4n}) \geq 4n-1 + W(K_{2n})$.
\end{theorem}

By combining these two results he obtained an even better lower bound on $W(K_{2n})$:
\begin{theorem}\label{tPetrosyan4npq}
If $n=p2^q$, where $p$ is odd, $q \in \mathbb{Z}_+$, then $W(K_{2n}) \geq 4n-2 - p - q$.
\end{theorem}

In that paper the second author also posed the following conjecture:
\begin{conjecture}
\label{conjPQ}
If $n=p2^q$, where $p$ is odd, $q \in \mathbb{Z}_+$, then $W(K_{2n}) = 4n-2 - p - q$.
\end{conjecture}
He verified this conjecture for $n \leq 4$, but the first author disproved it by constructing an interval $14$-coloring of $K_{10}$ in \cite{Khachatrian2012}. In ``Cycles and Colorings 2012'' workshop the second author presented another conjecture on $W(K_{2n})$:
\begin{conjecture}
\label{conjLog}
If $n \in \mathbb{N}$, then $W(K_{2n}) = 4n - 2 - \left\lfloor\log_2{n}\right\rfloor - \left \| n_2 \right \|$, where $\left \| n_2 \right \|$ is the number of $1$'s in the binary representation of $n$.
\end{conjecture}

In Section 2 we show that the problem of constructing an interval coloring of a complete graph $K_{2n}$ is equivalent to finding a special 1-factorization of the same graph. In Section 3 we use this equivalence to improve the lower bounds of Theorems \ref{tPetrosyan3n2} and \ref{tPetrosyan4n}, and disprove Conjecture \ref{conjLog}. Section 4 improves the upper bound of Theorem \ref{tGiaroKubale4n4} for complete graphs. In Section 5 we determine the exact values of $W(K_{2n})$ for $n \leq 12$ and improve Theorem \ref{tPetrosyan4npq}.

\section{From interval colorings to 1-factorizations}

Let the vertex set of a complete graph $K_{2n}$ be $V(K_{2n}) = \{u_i, v_i\ |\ i=1,2,\ldots,n\}$. For any fixed ordering of the vertices $\mathbf{v} = \left(u_1,v_1, u_2,v_2, \ldots,u_n,v_n\right)$ we denote by $H_{\mathbf{v}}^{[i,j]}$, $i \leq j$, the subgraph of $K_{2n}$ induced by the vertices $u_i, v_i, u_{i+1}, v_{i+1}, \ldots, u_j, v_j$. 

Let $\mathfrak{F} = \{F_1, F_2, \ldots, F_{2n-1}\}$ be a $1$-factorization of $K_{2n}$. For every $F\in \mathfrak{F}$ we define its \textit{left and right parts} with respect to the ordering of vertices $\mathbf{v}$:
\begin{align*}
l_{\mathbf{v}}^i(F) &= F \cap E\left(H_{\mathbf{v}}^{[1,i]}\right)\\
r_{\mathbf{v}}^i(F) &= F \cap E\left(H_{\mathbf{v}}^{[i+1,n]}\right)
\end{align*}

If for some $i$, $1 \leq i \leq n-1$, $F = l_{\mathbf{v}}^i(F) \cup r_{\mathbf{v}}^i(F)$ then $F$ is called an \textit{$i$-splitted} perfect matching with respect to the ordering $\mathbf{v}$. In other words the edges of $F$ do not cross the vertical line between the $i$-th and $(i+1)$-th pairs of vertices ($F_1^1$ and $F_1^2$ on Fig. \ref{K_6factorization}). 

Let $\alpha$ be any interval edge-coloring of $K_{2n}$. By renaming the vertices we can achieve the following inequalities: 
\begin{center}
$\underline{S}(u_i, \alpha) \leq \underline{S}(v_i, \alpha) \leq \underline{S}(u_{i+1}, \alpha)\leq \underline{S}(v_{i+1}, \alpha)$, $i=1,2,\ldots,n-1$.
\end{center}
So every coloring $\alpha$ implies a special ordering of vertices $\mathbf{v}_\alpha = \left(u_1,v_1, u_2,v_2, \ldots,u_n,v_n\right)$ for which these inequalities are satisfied. 

Now we fix the ordering $\mathbf{v}_\alpha$ and investigate some properties of the coloring $\alpha$. First we show that the spectrums of the vertices $u_i$ and $v_i$ are the same.

\begin{remark}
\label{spectrumIntersection}
For every $\alpha$ interval edge-coloring of $K_{2n}$, $S_\cap\left(K_{2n}, \alpha\right) \neq \emptyset$. Otherwise it would contradict the upper bound in Theorem \ref{upperKamalian}.
\end{remark}

\begin{lemma}
If $1 \leq i \leq n$, then $\underline{S}(u_i, \alpha) = \underline{S}(v_i, \alpha)$.
\end{lemma}
\begin{proof}
Remark \ref{spectrumIntersection} implies that if $\underline{S}(v_i, \alpha) - \underline{S}(u_{i}, \alpha) > 0$, then the edges colored by $\underline{S}(u_{i}, \alpha)$ form a perfect matching in the subgraph $K_{2n}\left[\{u_1,v_1,u_2,v_2,\ldots,u_i\}\right]$, which is impossible, as it has odd number of vertices.
\end{proof}

For the coloring $\alpha$ we define its \textit{shift vector} in the following way:
\begin{center}
${\rm sh}(\alpha) = (b_1, b_2, \ldots, b_{n-1})$ \\
where $b_i = \underline{S}(u_{i+1}, \alpha) - \underline{S}(u_i, \alpha)$, $i=1,2,\ldots,n-1$
\end{center}

By $B_i$ we denote the partial sums: $B_0 = 0$ and $B_i = \sum\limits_{j=1}^{i}b_j$, $i=1,2,\ldots,n-1$.

The \textit{total shift} of the coloring $\alpha$ is defined as follows:

\begin{center}
$|{\rm sh}(\alpha)| = B_{n-1} = \sum\limits_{i=1}^{n-1}b_i$
\end{center}

\begin{remark}
\label{totalShift}
If $\alpha$ is an interval $t$-coloring of $K_{2n}$ and ${\rm sh}(\alpha) = (b_1, b_2, \ldots, b_{n-1})$, then $t = 2n-1 + |{\rm sh}(\alpha)|$.
\end{remark}

\begin{remark}
\label{middleColors}
For every $\alpha$ interval edge-coloring of $K_{2n}$, the colors that appear in all vertices are $S_\cap\left(K_{2n}, \alpha\right) = \left[\underline{S}(u_n, \alpha), \overline{S}(u_1, \alpha)\right] = \left[|{\rm sh}(\alpha)|+1, 2n-1\right] = \left\{|{\rm sh}(\alpha)|+j \ |\ j=1,2,\ldots,2n-1-|{\rm sh}(\alpha)|\right\}$
\end{remark}

For every $i=1,2,\ldots,n-1$, we define the following two sets of colors:
\begin{align*}
L_{\mathbf{v}_\alpha}^i(\alpha) &= \left\{
\begin{tabular}{ll}
$[\underline{S}(u_{i}, \alpha), \underline{S}(u_{i+1}, \alpha) - 1] 
=\{ B_{i-1} + j \ |\ j=1,2,\ldots,b_i \}$, & if $b_i>0$,\\
$\emptyset$, & if $b_i=0$,\\
\end{tabular}%
\right.\\
R_{\mathbf{v}_\alpha}^i(\alpha) &= \left\{
\begin{tabular}{ll}
$[\overline{S}(u_{i}, \alpha) + 1, \overline{S}(u_{i+1}, \alpha)] 
=\{ B_{i-1} + 2n-1 + j \ |\ j=1,2,\ldots,b_i \}$, & if $b_i>0$,\\
$\emptyset$, & if $b_i=0$.\\
\end{tabular}%
\right.\\
\end{align*}

\begin{remark}
\label{splittedColors}
If $\alpha$ is an interval $t$-coloring of $K_{2n}$ and ${\rm sh}(\alpha) = (b_1, b_2, \ldots, b_{n-1})$, then
\begin{center}
\begin{tabular}{ll}
$L_{\mathbf{v}_\alpha}^i(\alpha) \subset S_\cap\left(H_{\mathbf{v}_\alpha}^{[1,i]},\alpha\right)$
&$L_{\mathbf{v}_\alpha}^i(\alpha) \cap S_\cup\left(H_{\mathbf{v}_\alpha}^{[i+1,n]},\alpha\right) = \emptyset$\\
$R_{\mathbf{v}_\alpha}^i(\alpha) \cap S_\cup\left(H_{\mathbf{v}_\alpha}^{[1,i]},\alpha\right) = \emptyset$ 
&$R_{\mathbf{v}_\alpha}^i(\alpha) \subset S_\cap\left(H_{\mathbf{v}_\alpha}^{[i+1,n]},\alpha\right)$
\end{tabular}
\end{center}
\end{remark}

By $C_k(\alpha)$ we denote the edges colored by the color $k$: $C_k(\alpha) = \{e \in E(K_{2n})\ |\ \alpha(e)=k\}$.

\begin{figure}[t!]
\centering
\includegraphics[width=0.7\textwidth]{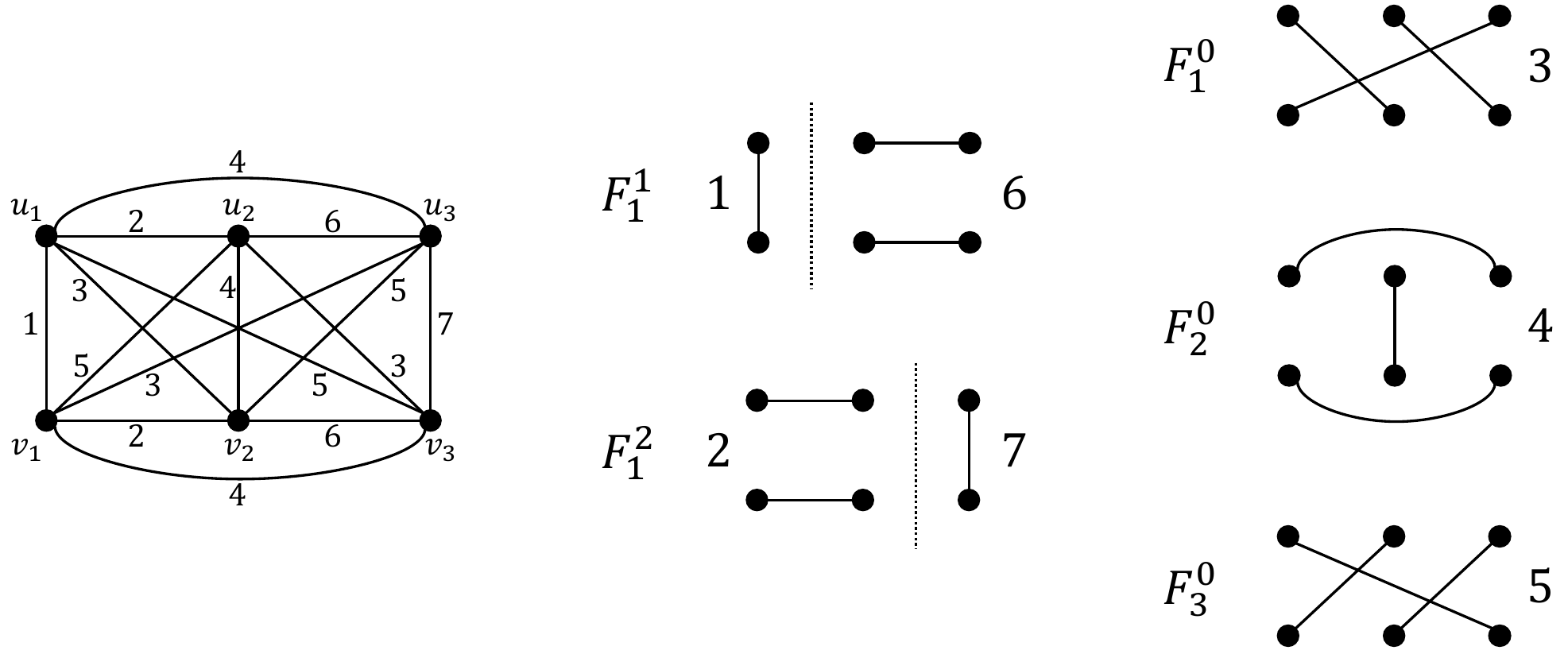}
\caption{Interval $7$-coloring of $K_6$ and the corresponding 1-factorization $\mathfrak{F}=\left\{F_1^1, F_1^2, F_1^0, F_2^0, F_3^0\right\}$}
\label{K_6factorization}
\end{figure}

\begin{lemma}[Equivalence lemma]\label{lEquiv}
The following two statements are equivalent:
\begin{description}
\item{(a)} there exists $\alpha$ interval edge-coloring of $K_{2n}$  such that ${\rm sh}(\alpha) = (b_1, b_2, \ldots b_{n-1})$,
\item{(b)} there exist $\mathbf{v}$ ordering of vertices and $\mathfrak{F} = 
\left\{ F^0_j\ |\ j=1,2,\ldots,2n-1-\sum\limits_{i=1}^{n-1}b_i \right\}
\cup
\bigcup\limits_{i=1}^{n-1}\left\{F^i_j\ |\ j=1,2,\ldots,b_i\right\}$ 1-factorization of $K_{2n}$ such that $F^i_j$ is $i$-splitted with respect to the ordering $\mathbf{v}$, $i=1,2,\ldots,n-1$, $j=1,2,\ldots,b_i$, $b_i \in \mathbb{Z}_+$.
\end{description}
\end{lemma}

\begin{proof}
Throughout the proof we will use $B_i$ as a shorthand for $\sum\limits_{j=1}^{i}b_j$, $i=0,1,\ldots,n-1$.

\begin{description}
\item{(a) $=>$ (b)}. Let $\alpha$ be an interval $t$-coloring of $K_{2n}$ such that ${\rm sh}(\alpha) = (b_1, b_2, \ldots b_{n-1})$. We choose the ordering $\mathbf{v}_\alpha$ and construct the 1-factorization $\mathfrak{F}$ of $K_{2n}$. 

According to Remark \ref{middleColors}, there exist $2n-1-|{\rm sh}(\alpha)|$ colors that appear in the spectrums of all the vertices. By definition, $|{\rm sh}(\alpha)| = \sum\limits_{i=1}^{n-1}b_i$, so we take $F_j^0 = C_{|{\rm sh}(\alpha)|+j}(\alpha)$, for every $j=1,2,\ldots,2n-1-|{\rm sh}(\alpha)|$.

For every $i=1,2,\ldots,n-1$, Remark \ref{splittedColors} implies there exist $|L_{\mathbf{v}_\alpha}^i(\alpha)|=b_i$ distinct colors that appear only in the spectrums of the first $i$ pairs of vertices and another $|R_{\mathbf{v}_\alpha}^i(\alpha)|=b_i$ distinct colors that appear only in the spectrums of the remaining $2n-2i$ vertices. We take $F^i_j = C_{B_{i-1} + j}(\alpha) \cup C_{B_{i-1} + 2n - 1 + j}(\alpha)$, for every $i=1,2,\ldots,n-1$ and $j=1,2,\ldots,b_i$. Note that the edges colored by the colors from $L_{\mathbf{v}_\alpha}^i(\alpha) \cup R_{\mathbf{v}_\alpha}^i(\alpha)$ do not cross the vertical line between the $i$-th and $(i+1)$-th pairs of vertices ($F_1^1$ and $F_1^2$ on Fig. \ref{K_6factorization}), so $F^i_j$ is $i$-splitted with respect to the ordering $\mathbf{v}_\alpha$ for all permitted $j$.

\item{(b) $=>$ (a)}. Suppose $\mathfrak{F} = 
\left\{ F^0_j\ |\ j=1,2,\ldots,2n-1-|{\rm sh}(\alpha)| \right\}
\cup
\bigcup\limits_{i=1}^{n-1}\left\{F^i_j\ |\ j=1,2,\ldots,b_i\right\}$ is a 1-factorization of $K_{2n}$ with the property that $F_j^i$ is $i$-splitted perfect matching with respect to the ordering $\mathbf{v}=\left(u_1,v_1, u_2,v_2, \ldots,u_n,v_n\right)$, $i=1,2,\ldots,n-1$, $j=1,2,\ldots,b_i$. We construct $\alpha$ interval edge-coloring of $K_{2n}$ in the following way:

\begin{tabular}{lll}
$\alpha(e)=B_{i-1} + j$ & if $e \in l_{\mathbf{v}}^i(F_j^i)$ & $i=1,2,\ldots,n-1$, $j=1,2,\ldots,b_i$ \\
$\alpha(e)=B_{n-1} + j$ & if $e \in F_j^0$ & $j=1,2,\ldots,2n-1-B_{n-1}$\\
$\alpha(e)=B_{i-1} + 2n - 1 + j$ & if $e \in r_{\mathbf{v}}^i(F_j^i)$ & $i=1,2,\ldots,n-1$, $j=1,2,\ldots,b_i$
\end{tabular}

The fact that $F^i_j$ is $i$-splitted with respect to the ordering $\mathbf{v}$ implies that every edge of $K_{2n}$ have received a color. The vertex $u_i$ (also $v_i$) is covered by all perfect matchings $F_j^0$, $j=1,2,\ldots,2n-1-B_{n-1}$, by the left parts of the matchings $F_j^{i'}$, $i'=i,i+1,\ldots,n-1$, and by the right parts of the matchings $F_j^{i'}$, $i'=1,2,\ldots,i-1$, for every $j=1,2,\ldots,b_{i'}$. So the spectrum is:
\begin{align*}
S(u_i, \alpha) = S(v_i, \alpha) &= \bigcup\limits_{i'=i}^{n-1}\{B_{i'-1} + j \ |\ j=1,2,\ldots,b_{i'}\} \\
& \cup
\{B_{n-1} + j\ |\ j=1,2,\ldots,2n-1-B_{n-1}\} \\
& \cup
\bigcup\limits_{i'=1}^{i-1}\{B_{i'-1} + 2n-1 + j \ |\ j=1,2,\ldots,b_{i'}\}\\
&= [B_{i-1}+1, B_{n-1}] \cup [B_{n-1}+1, 2n-1] \cup [2n, B_{i-1}+2n-1]\\
&= [B_{i-1}+1, B_{i-1}+2n-1]
\end{align*}

This proves that $\alpha$ is an interval $(B_{n-1} + 2n-1)$-coloring of $K_{2n}$. To complete the proof of the lemma we need to check the shift vector of the coloring $\alpha$. Note that for every $i=1,2,\ldots,n-1$, we have $\underline{S}(u_{i+1}, \alpha) - \underline{S}(u_{i}, \alpha) = B_{i}-B_{i-1} = b_i$. This shows that the ordering $\mathbf{v}_\alpha$ coincides with the ordering $\mathbf{v}$ and ${\rm sh}(\alpha) = (b_1, b_2, \ldots, b_{n-1})$.
\end{description}

\end{proof}

\begin{remark}\label{splittedSameColor}
Some of the matchings $F_j^0$ constructed in the first part of the proof of Equivalence lemma may be splitted perfect matchings as well, but for each of them both their left and right parts have the same color in the coloring $\alpha$. For example, in case $|{\rm sh}(\alpha)|=0$, $F^0_{\alpha(u_1v_1)} = C_{\alpha(u_1v_1)}(\alpha)$ is $1$-splitted perfect matching with respect to the ordering $\mathbf{v}_\alpha$.
\end{remark}

\begin{corollary}\label{cEquiv}
For any $n\in\mathbb{N}$, $K_{2n}$ has an interval $t$-coloring if and only if it has a 1-factorization, where at least $t-2n+1$ perfect matchings are splitted.
\end{corollary}
\begin{proof}
Construction of the desired 1-factorization from the interval $t$-coloring immediately follows from Remark \ref{totalShift} and Equivalence lemma.  Remark \ref{splittedSameColor} implies that the number of the splitted perfect matchings in the obtained 1-factorization can be more than $t-2n+1$.

If we have a 1-factorization of $K_{2n}$ with at least $t-2n+1$ splitted perfect matchings we can arbitrarily choose exactly $t-2n+1$ of them, then for each of them choose the $i$ for which it is $i$-splitted (the same perfect matching can be both $i$-splitted and $i'$-splitted for distinct $i$ and $i'$, the choice is again arbitrary) and apply Equivalence lemma. So, the corresponding coloring may not be uniquely determined.
\end{proof}

This corollary shows that finding an interval edge-coloring of $K_{2n}$ with many colors is equivalent to finding a 1-factorization with many splitted perfect matchings with respect to some ordering of vertices. For the ordering $\mathbf{v}$ we can define the maximum number of splitted perfect matchings over all 1-factorizations of $K_{2n}$. Because of the symmetry of complete graph this number does not actually depend on the chosen ordering $\mathbf{v}$, so we denote it by $\sigma_n$.

\begin{theorem}[Equivalence theorem]\label{tEquiv}
For every $n\in \mathbb{N}$, $W(K_{2n}) = 2n-1+\sigma_n$. 
\end{theorem}

\section{Lower bounds}

In order to obtain new lower bounds on $W(K_{2n})$ we split $K_{2n}$ into two edge-disjoint spanning regular subgraphs, find convenient 1-factorizations for each of them, and then apply Equivalence theorem for the union of these 1-factorizations.

We fix the ordering of vertices of $K_{2n}$, $\mathbf{v} = \left(u_1,v_1, u_2,v_2, \ldots,u_n,v_n\right)$, and define two spanning regular subgraphs of $K_{2n}$, $K_2 \square K_n$ and $K_2 \times K_n$ (Fig. \ref{K_8products}):
\begin{align*}
&V(K_2 \square K_n) = V(K_2 \times K_n) = V(K_{2n})\\
&E(K_2 \square K_n) = \{u_iu_j\ |\ 1\leq i<j \leq n\} \cup \{u_iv_i\ |\ 1\leq i \leq n\} \cup \{v_iv_j\ |\ 1\leq i<j \leq n\}\\
&E(K_2 \times K_n) = \{u_iv_j\ |\ 1\leq i \neq j \leq n\}
\end{align*}

\begin{figure}[t!]
\centering
\includegraphics[width=0.43\textwidth]{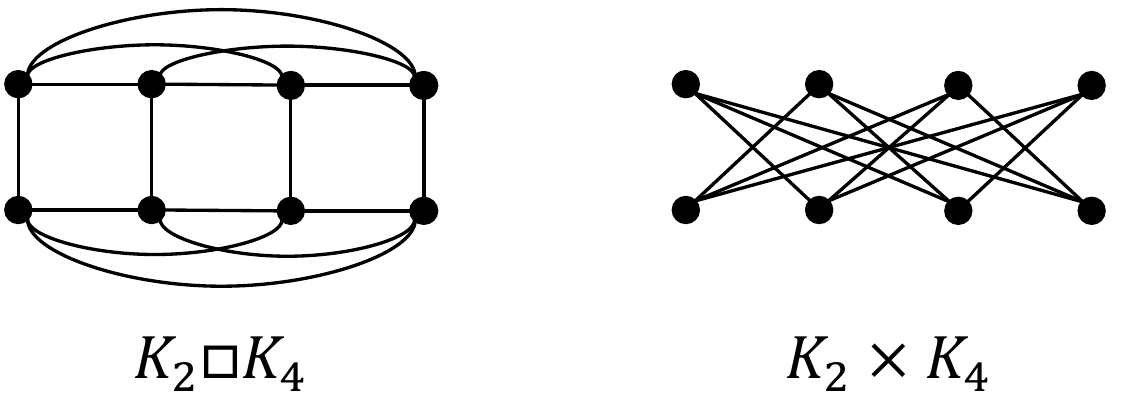}
\caption{Two spanning regular subgraphs of $K_8$}
\label{K_8products}
\end{figure}

Note that $E(K_{2n})=E(K_2 \square K_n)\cup E(K_2 \times K_n)$. We fix an ordering of vertices $\mathbf{v} = \left(u_1,v_1, u_2,v_2, \ldots,u_n,v_n\right)$ and define a special 1-factorization of $K_2 \square K_n$ which we denote by $\mathfrak{P}_n$:

$\mathfrak{P}_n = \{P_0, P_1, \ldots, P_{n-1}\}$, where 
\begin{align*}
P_0 &= \left\{
\begin{tabular}{ll}
$\{u_ju_{n+1-j}, v_jv_{n+1-j}\ |\ j=1,2,\ldots,\frac{n}{2}\}$ 
& if $n$ is even\\
$\{u_ju_{n+1-j}, v_jv_{n+1-j}\ |\ j=1,2,\ldots,\lfloor\frac{n}{2}\rfloor\}
\cup \{u_{\frac{n+1}{2}}v_{\frac{n+1}{2}}\}$, & if $n$ is odd\\
\end{tabular}%
\right.
\end{align*}

For every $i=1,2,\ldots,n-1$, $P_i = l_{\mathbf{v}}^i(P_i) \cup r_{\mathbf{v}}^i(P_i)$, where
\begin{align*}
l_{\mathbf{v}}^i(P_i) &= \left\{
\begin{tabular}{ll}
$\{u_ju_{i+1-j}, v_jv_{i+1-j}\ |\ j=1,2,\ldots,\frac{i}{2}\}$ 
& if $i$ is even\\
$\{u_ju_{i+1-j}, v_jv_{i+1-j}\ |\ j=1,2,\ldots,\lfloor\frac{i}{2}\rfloor\}
\cup \{u_{\frac{i+1}{2}}v_{\frac{i+1}{2}}\}$, & if $i$ is odd\\
\end{tabular}%
\right.\\
r_{\mathbf{v}}^i(P_i) &= \left\{
\begin{tabular}{ll}
$\{u_{i+j}u_{n+1-j}, v_{i+j}v_{n+1-j}\ |\ j=1,2,\ldots,\frac{n-i}{2}\}$ 
& if $n-i$ is even\\
$\{u_{i+j}u_{n+1-j}, v_{i+j}v_{n+1-j}\ |\ j=1,2,\ldots,\lfloor\frac{n-i}{2}\rfloor\}
\cup \{u_{\frac{n+i+1}{2}}v_{\frac{n+i+1}{2}}\}$, & if $n-i$ is odd\\
\end{tabular}%
\right.
\end{align*}

\begin{figure}[t!]
\centering
\includegraphics[width=0.7\textwidth]{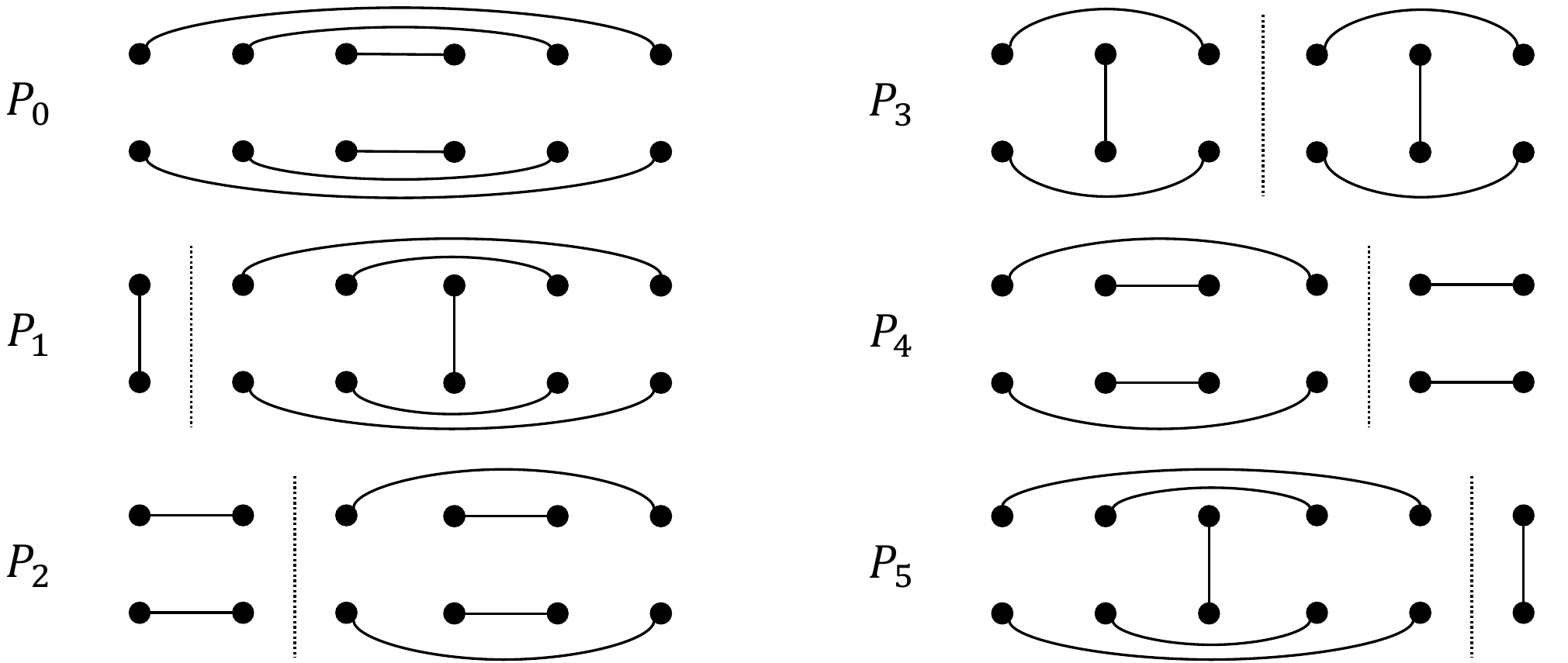}
\caption{1-factorization $\mathfrak{P}_6$ of $K_2 \square K_6$}
\label{P_6}
\end{figure}

$P_i$ is clearly an $i$-splitted perfect matching, for every $i=1,2,\ldots,n-1$. Note, that $K_2 \times K_n$ is a regular bipartite graph, so König's theorem \cite{Konig1916} implies it has a 1-factorization. If we consider the perfect matchings of any 1-factorization of $K_2 \times K_n$ as non-splitted matchings and add the perfect matchings of $\mathfrak{P}_n$ we obtain that $\sigma_n \geq n-1$. Equivalence theorem implies that this result is equivalent to Theorem \ref{tPetrosyan3n2}. 

In order to improve this bound we concentrate on finding a better 1-factorization of $K_2 \times K_n$.

\begin{lemma}\label{l35n3}
If $n \geq 2$, then $\sigma_n \geq \lfloor 1.5n \rfloor - 2$.
\end{lemma}
\begin{proof}
We fix an ordering of vertices $\mathbf{v} = \left(u_1,v_1,u_2,v_2,\ldots,u_n,v_n\right)$ and consider two induced subgraphs:
\begin{align*}
G_1 &= K_2 \times K_n\left[\left\{u_1,v_1,u_2,v_2,\ldots,u_{\lfloor\frac{n}{2}\rfloor},v_{\lfloor\frac{n}{2}\rfloor}\right\}\right]\\ 
G_2 &= K_2 \times K_n\left[\left\{u_{\lfloor\frac{n}{2}\rfloor + 1},v_{\lfloor\frac{n}{2}\rfloor + 1},
u_{\lfloor\frac{n}{2}\rfloor + 2},v_{\lfloor\frac{n}{2}\rfloor + 2},\ldots,u_n,v_n \right\}\right]
\end{align*}
Both subgraphs are regular and bipartite, so according to the König's theorem \cite{Konig1916} they have 1-factorizations. Let the 1-factorizations of $G_1$ and $G_2$ be $F^l_1,F^l_2,\ldots,F^l_{\lfloor\frac{n}{2}\rfloor-1}$ and $F^r_1,F^r_2,\ldots,F^r_{\lceil\frac{n}{2}\rceil-1}$, respectively. By joining the first $\lfloor\frac{n}{2}\rfloor-1$ pairs of these matchings we form $\lfloor\frac{n}{2}\rfloor$-splitted perfect matchings of $K_2 \times K_n$ with respect to the ordering $\mathbf{v}$:

\begin{center}
$F_i = F^l_i \cup F^r_i$, for all $i=1,2,\ldots,\lfloor\frac{n}{2}\rfloor - 1$.
\end{center}

If we remove the edges $\bigcup\limits_{i=1}^{\lfloor\frac{n}{2}\rfloor-1}F_i$ from the graph $K_2 \times K_n$, the remaining graph is still a regular bipartite graph and has a 1-factorization, which we denote by $\mathfrak{F}_0$. Now, $\mathfrak{F}_0 \cup \bigcup\limits_{i=1}^{\lfloor\frac{n}{2}\rfloor-1}F_i \cup \mathfrak{P}_n$ is a 1-factorization of $K_{2n}$. The number of splitted matchings is $\lfloor\frac{n}{2}\rfloor - 1 + n - 1$. So we have $\sigma_n \geq \lfloor 1.5n \rfloor - 2$.
\end{proof}

By applying Equivalence theorem we obtain the following lower bound:
\begin{theorem}
\label{t35n3}
If $n \geq 2$, then $W(K_{2n}) \geq \lfloor 3.5n \rfloor - 3$.
\end{theorem}

This theorem implies that $W(K_{10}) \geq 14$ which is the smallest example that disproves Conjecture \ref{conjPQ}. Next we focus on the case when $n$ is a composite number.

\begin{lemma}\label{lComposite}
For any $m,n \in\mathbb{N}$, $\sigma_{mn} \geq \sigma_m + \sigma_n + 2(m-1)(n-1)$.
\end{lemma}
\begin{proof}
Let the vertex sets of $K_{2mn}$, $K_{2n}$ and $K_{2m}$ be as follows:
\begin{align*}
V(K_{2mn}) &= \left\{u_i^j,v_i^j\ |\ i=1,2,\ldots,n,\ j=1,2,\ldots,m\right\} \\
V(K_{2n}) &= \left\{\overline{u}_i,\overline{v}_i\ |\ i=1,2,\ldots,n \right\} \\ 
V(K_{2m}) &= \left\{\widetilde{u}^i,\widetilde{v}^i\ |\ i=1,2,\ldots,m \right\} 
\end{align*}

\begin{figure}[t!]
\centering
\includegraphics[width=0.39\textwidth]{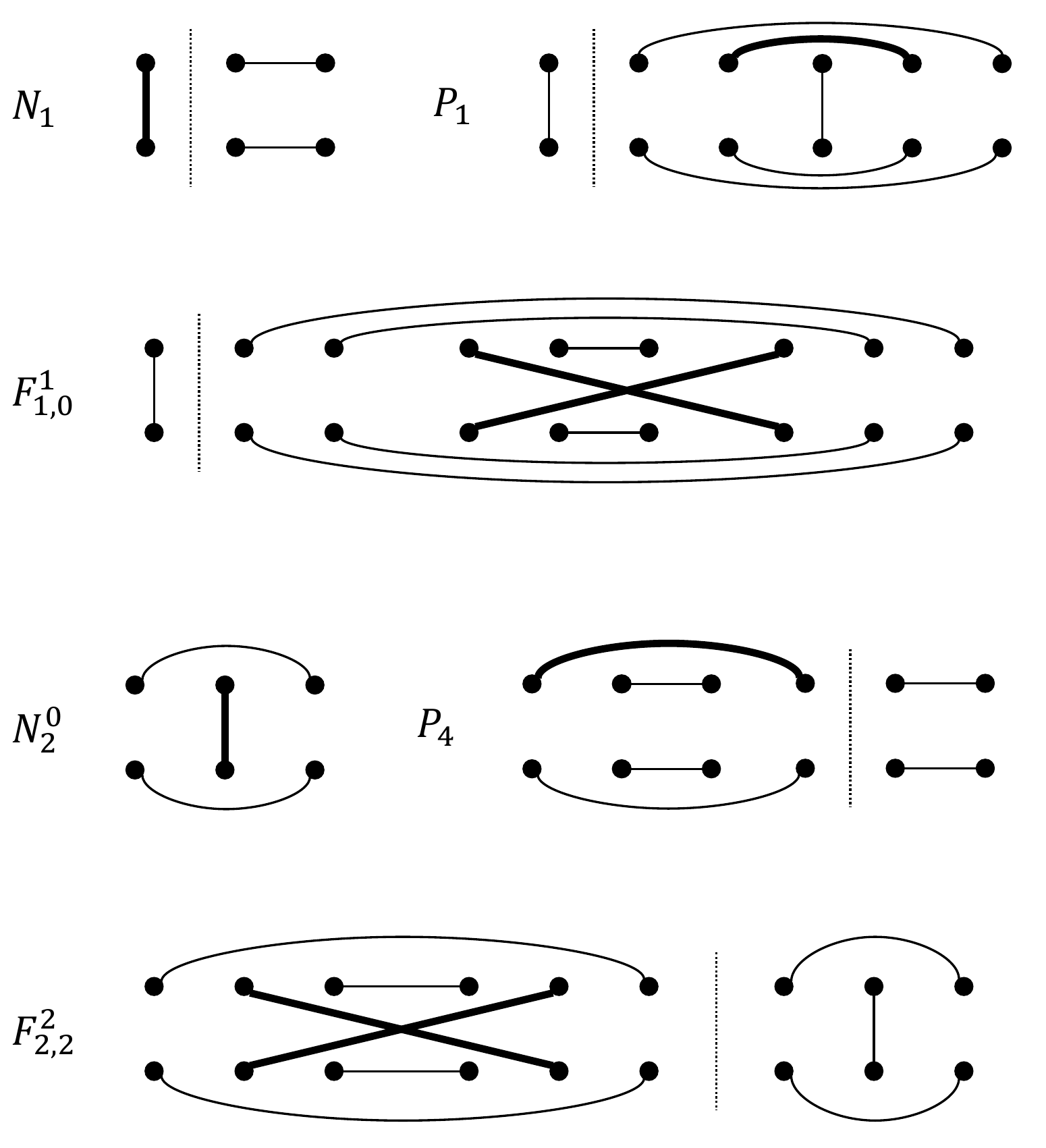}
\hspace{1cm}
\includegraphics[width=0.39\textwidth]{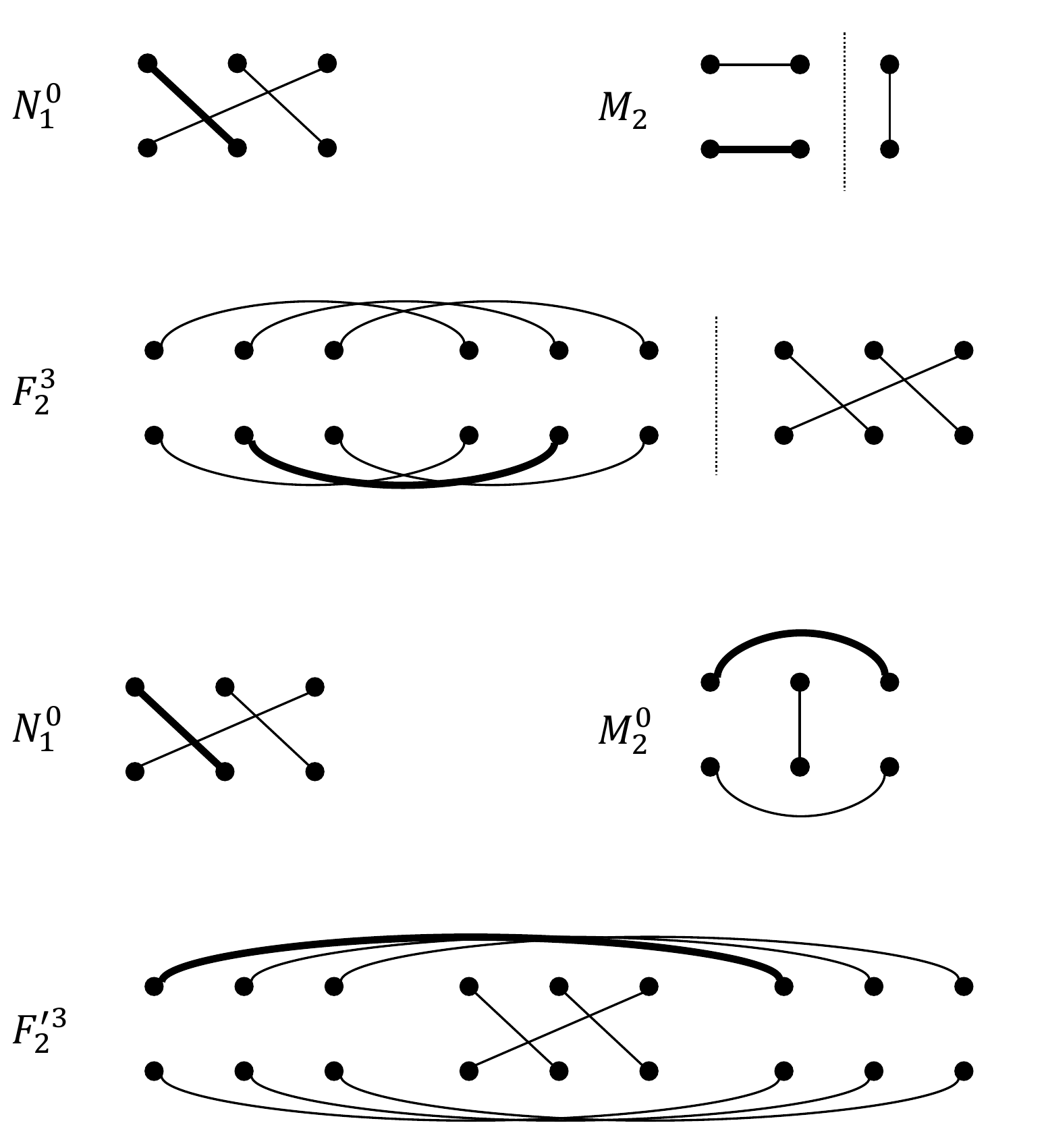}
\caption{Several perfect matchings of $K_{18}$ constructed based on 1-factorizations $\overline{\mathfrak{F}}=\{N_1,N_2,N_1^0,N_2^0,N_3^0\}$ of $K_6$, $\widetilde{\mathfrak{F}}=\{M_1,M_2,M_1^0,M_2^0,M_3^0\}$ of $K_6$ and $\mathfrak{P}_6=\{P_0,P_1,P_2,P_3,P_4,P_5\}$ of $K_2\square K_6$ using Lemma \ref{lComposite}}
\label{K_18matchings}
\end{figure}

We fix the following orderings of vertices of $K_{2mn}$, $K_{2n}$ and $K_{2m}$, respectively:
\begin{align*}
\mathbf{v} &= \left(
u^1_1,v^1_1,u^1_2,v^1_2, \ldots, u^1_n,v^1_n,
u^2_1,v^2_1,u^2_2,v^2_2, \ldots, u^2_n,v^2_n, \ldots, 
u^m_1,v^m_1,u^m_2,v^m_2, \ldots, u^m_n,v^m_n \right) \\
\overline{\mathbf{v}} &= \left(
\overline{u}_1,\overline{v}_1,\overline{u}_2,\overline{v}_2, \ldots, \overline{u}_n,\overline{v}_n\right)\\
\widetilde{\mathbf{v}} &= \left(
\widetilde{u}^1,\widetilde{v}^1,\widetilde{u}^2,\widetilde{v}^2, \ldots, \widetilde{u}^m,\widetilde{v}^m\right)
\end{align*}

Let $\overline{\mathfrak{F}} = \{ N_1,N_2,\ldots,N_{\sigma_n},N^0_1,N^0_2,\ldots, N^0_{2n-1-\sigma_n} \}$ be a 1-factorization of $K_{2n}$, where $N_i$, $i=1,2,\ldots,\sigma_n$ are splitted perfect matchings. Let $\widetilde{\mathfrak{F}} = \{ M_1,M_2,\ldots,M_{\sigma_m},M^0_1,M^0_2,\ldots, M^0_{2m-1-\sigma_m} \}$ be a 1-factorization of $K_{2m}$, where $M_i$, $i=1,2,\ldots,\sigma_m$ are splitted perfect matchings. 

We also need the graph $K_2 \square K_{2m}$ with the vertex set $\left\{w_i,z_i\ |\ i=1,2,\ldots,2m\right\}$, an ordering of its vertices $\mathbf{w} = \left( w_1,z_1,w_2,z_2,\ldots,w_{2m},z_{2m} \right)$, and its 1-factorization $\mathfrak{P}_{2m} = \{P_0,P_1,\ldots,P_{2m-1} \}$ as defined at the beginning of this section. We call the subgraph $K_2 \square K_{2m}[\{w_{2k-1},w_{2k},z_{2k-1},z_{2k}\}]$ $k$-th cell of $K_2 \square K_{2m}$, $1 \leq k \leq m$.

During the proof we always assume that $x,y \in \{u,v\}$, $1 \leq s,t \leq n$ and $1 \leq p,q \leq m$.

Let $\overline{\varphi}$ be a mapping which projects the edges of $K_{2mn}$ to the edges of $K_{2n}$. For every edge $x^p_sy^q_t \in E(K_{2mn})$, where $x_s \neq y_t$, we define $\overline{\varphi}(x^p_sy^q_t)=\overline{x}_s\overline{y}_t$. Next we define a mapping $\widetilde{\varphi}$ which projects the remaining edges of $K_{2mn}$ to the edges of $K_{2m}$. For every edge $x^p_sx^q_s \in E(K_{2mn})$ we define $\widetilde{\varphi}( x^p_sx^q_s ) = \widetilde{x}^p\widetilde{x}^q$. Note that the preimages $\overline{\varphi}^{-1}(\overline{e})$ for all $\overline{e} \in E(K_{2n})$ and $\widetilde{\varphi}^{-1}(\widetilde{x}^p\widetilde{x}^q)$ for all $\widetilde{x}^p\widetilde{x}^q \in E(K_{2m})$ are pairwise disjoint and their union covers the set $E(K_{2mn})$. We split the edge set $E(K_{2mn})$ into three parts the following way:
\begin{align*}
E(K_{2mn}) &= E^1 \cup E^2 \cup E^3 \text{, where} \\
E^1 &= \bigcup\limits_{i=1}^{\sigma_n}{\bigcup\limits_{\overline{e} \in N_i}{\overline{\varphi}^{-1}(\overline{e})}} \\
E^2 &= \bigcup\limits_{i=2}^{2n-1-\sigma_n}\bigcup\limits_{\overline{e} \in N^0_i}{\overline{\varphi}^{-1}(\overline{e})} \\
E^3 &= \bigcup\limits_{\overline{e} \in N^0_1}{\overline{\varphi}^{-1}(\overline{e})} \cup \bigcup\limits_{\widetilde{x}^p\widetilde{x}^q \in E(K_{2m})}{\widetilde{\varphi}^{-1}(\widetilde{x}^p\widetilde{x}^q)}
\end{align*}

The 1-factorization of $K_{2mn}$ we are going to construct is denoted by $\mathfrak{F}$ and also consists of three parts. 

\begin{center}
$\mathfrak{F} = \mathfrak{F}^1 \cup \mathfrak{F}^2 \cup \mathfrak{F}^3$
\end{center}

The set of perfect matchings $\mathfrak{F}^k$ covers the set $E^k$, $k=1,2,3$. Fig. \ref{K_18matchings} displays example perfect matchings for each of the parts in case $m=n=3$. 

The set $E^1$ contains the preimages of splitted perfect matchings of $K_{2n}$. To cover it, for every splitted perfect matching $N_i \in \overline{\mathfrak{F}}$, $i=1,2,\ldots,\sigma_n$, and for every perfect matching with an odd index $P_{2j+1} \in \mathfrak{P}_{2m}$, $j=0,1,\ldots,m-1$, we construct one perfect matching of $\mathfrak{F}^1$. 
\begin{align*}
F^1_{i,j} = &F^1_{i,j,1} \cup F^1_{i,j,2} \cup F^1_{i,j,3} \cup F^1_{i,j,4}\text{, where }\\
F^1_{i,j,1} = &\bigcup\limits_{\substack{w_{2k-1}z_{2k-1} \in P_{2j+1} \\ 1 \leq k \leq m}}
\left\{x_s^ky_t^k\ |\ \overline{x}_s\overline{y}_t \in l(N_i)\right\} \\
F^1_{i,j,2} = &\bigcup\limits_{\substack{w_{2k}z_{2k} \in P_{2j+1} \\ 1 \leq k \leq m}}
\left\{x_s^ky_t^k\ |\ \overline{x}_s\overline{y}_t \in r(N_i)\right\} \\
F^1_{i,j,3} = &\bigcup\limits_{\substack{w_{2k-1}w_{2l-1} \in P_{2j+1} \\ 1 \leq k < l \leq m}}
\left\{x_s^ky_t^l, y_t^kx_s^l\ |\ \overline{x}_s\overline{y}_t \in l(N_i)\right\} \\
F^1_{i,j,4} = &\bigcup\limits_{\substack{w_{2k}w_{2l} \in P_{2j+1} \\ 1 \leq k < l \leq m}}
\left\{x_s^ky_t^l, y_t^kx_s^l\ |\ \overline{x}_s\overline{y}_t \in r(N_i)\right\}\\
\mathfrak{F}^1 = &\left\{F^1_{i,j}\ |\ i=1,2,\ldots,\sigma_n,\ j=0,1,\ldots,m-1\right\}
\end{align*}

For $F^1_{i,j,1}$ and $F^1_{i,j,2}$, we look for vertical edges in $P_{2j+1}$. If for some $k$, the vertical edge of the left (right) part of the $k$-th cell belongs to $P_{2j+1}$, we add the preimages of all edges of $l(N_i)$ ($r(N_i)$) in the $k$-th copy of $K_{2n}$ in $K_{2mn}$ to $F^1_{i,j,1}$ ($F^1_{i,j,2}$). Every matching $P_{2j+1}$ contains exactly two vertical edges ($w_{j+1}z_{j+1}$ and $w_{j+m+1}z_{j+m+1}$). If $m$ is odd, then one of these two belongs to the left part of its cell, and the other one belongs to the right part of its cell. If $m$ is even, then if $j$ is odd (even), both vertical edges belong to the right (left) parts of the cells. So, the number of edges in $F^1_{i,j,1}$ and $F^1_{i,j,2}$ can be calculated the following way:
\begin{align*}
|F^1_{i,j,1}| = &|l(N_i)|\left( (m \bmod 2)\cdot 1 + (1 - m \bmod 2)\cdot 2(1 - j \bmod 2) \right) \\
|F^1_{i,j,2}| = &|r(N_i)|\left( (m \bmod 2)\cdot 1 + (1 - m \bmod 2)\cdot 2(j \bmod 2) \right) 
\end{align*}

For $F^1_{i,j,3}$ ($F^1_{i,j,4}$) we are looking for edges joining left side (right side) vertices of two different cells in $P_{2j+1}$. If $m$ is odd, then there are $\frac{m-1}{2}$ such edges. If $m$ is even, then there are $\frac{m}{2} - (1 - j \bmod 2)$ (in case of $F^1_{i,j,4}$: $\frac{m}{2} - (j \bmod 2)$) such edges. For every such edge which joins the $k$-th and $l$-th cells ($k<l$) we add the preimages of all edges in $l(N_i)$ ($r(N_i)$) which join the vertices in $k$-th and $l$-th copies of $K_{2n}$ in $K_{2mn}$ to $F^1_{i,j,3}$ ($F^1_{i,j,4}$). Note that for every chosen edge from $P_{2j+1}$, every edge in $l(N_i)$ ($r(N_i)$) has exactly 2 preimages in $F^1_{i,j,3}$ ($F^1_{i,j,4}$). So we have:
\begin{align*}
|F^1_{i,j,3}| = &2|l(N_i)|\left( (m \bmod 2)\cdot \frac{m-1}{2} + (1 - m \bmod 2)\cdot \left(\frac{m}{2} - (1 - j \bmod 2)\right) \right) \\
|F^1_{i,j,4}| = &2|r(N_i)|\left( (m \bmod 2)\cdot \frac{m-1}{2} + (1 - m \bmod 2)\cdot \left(\frac{m}{2} - (j \bmod 2)\right) \right)
\end{align*}

The construction of $F^1_{i,j}$ implies that it is a matching in $K_{2mn}$. To prove that it is also a perfect matching, we need to show that it has exactly $mn$ edges. 
\begin{align*}
|F^1_{i,j}| = &|F^1_{i,j,1}| + |F^1_{i,j,2}| + |F^1_{i,j,3}| + |F^1_{i,j,4}| = \\
 = &|l(N_i)|\left(
 	(m \bmod 2)( 1 + m - 1) + 
    (1 - m \bmod 2)\left( 
    	2(1- j \bmod 2) + m - 2(1 - j\bmod 2) 
	\right)
\right) + \\
 + &|r(N_i)|\left(
 	(m \bmod 2)( 1 + m - 1) + 
    (1 - m \bmod 2)\left(
    	2(j \bmod 2) + m - 2(j\bmod 2)
    \right)
\right) \\
 = &\left(|l(N_i)| + |r(N_i)|\right)\left((m \bmod 2)\cdot m +  (1 - m \bmod 2)\cdot m \right) = nm
\end{align*}

The matchings $F^1_{i,j}$ and $F^1_{i',j'}$ are disjoint if $i \ne i'$ or $j \ne j'$, as their edges correspond to either different edges in $K_{2n}$ or to different edges in $K_2 \square K_{2m}$. Also note that, if $N_i$ is an $r$-splitted matching for $\overline{\mathbf{v}}$, then $F^1_{i,j}$ is $(jn+r)$-splitted matching for $\mathbf{v}$, for every $i=1,2,\ldots,\sigma_n$ and $j=0,1,\ldots,m-1$.

The set $E^2$ contains the preimages of all but one non-splitted perfect matchings. To cover it, for every non-splitted perfect matching $N^0_i \in \overline{\mathfrak{F}}$ except $N^0_1$ (the choice of this exception is arbitrary) and for every perfect matching with an even index $P_{2j} \in \mathfrak{P}_{2m}$ we construct one perfect matching of $\mathfrak{F}^2$.
\begin{align*}
F^2_{i,j} = &F^2_{i,j,1} \cup F^2_{i,j,2}\text{, where }\\
F^2_{i,j,1} = &\bigcup\limits_{\substack{w_{2k-1}w_{2k} \in P_{2j} \\ 1 \leq k \leq m}}
\left\{x_s^ky_t^k\ |\ \overline{x}_s\overline{y}_t \in N^0_i\right\} \\
F^2_{i,j,2} = &\bigcup\limits_{\substack{w_{2k-1}w_{2l} \in P_{2j} \\ 1 \leq k < l \leq m}}
\left\{x_s^ky_t^l, y_t^kx_s^l\ |\ \overline{x}_s\overline{y}_t \in N^0_i\right\}\\
\mathfrak{F}^2 = &\{F^2_{i,j}\ |\ i=2,3,\ldots,2n-1-\sigma_n,\ j=0,1,\ldots,m-1\}
\end{align*}
The matchings $P_{2j}$ have only horizontal edges. We look for those edges which join a vertex from the left part of a cell to a vertex from the right part of a (possibly different) cell. If both endpoints of an edge belong to the same $k$-th cell, we add the preimages of all edges of $N_i^0$ which belong to the $k$-th copy of $K_{2n}$ in $K_{2mn}$ to the set $F^2_{i,j,1}$. The number of such edges in $P_{2j}$ is $1$ if $m$ is odd and $2(j \bmod 2)$ if $m$ is even. So we have:
\begin{align*}
|F^2_{i,j,1}| = &n\left( (m \bmod 2)\cdot 1 + (1 - m \bmod 2)\cdot 2(j \bmod 2) \right)
\end{align*}
If the edge of $P_{2j}$ joins vertices of $k$-th and $l$-th cells ($k < l$) then we add both preimages of all edges of $N_i^0$ which join the vertices of $k$-th and $l$-th copies of $K_{2n}$ in $K_{2mn}$ to $F^2_{i,j,2}$. The number of such edges in $P_{2j}$ is $\frac{m-1}{2}$ if $m$ is odd, and $\frac{m}{2} - (j \bmod 2)$ if $m$ is even. So,
\begin{align*}
|F^2_{i,j,2}| = &2n\left( (m \bmod 2)\cdot \frac{m-1}{2} + (1 - m \bmod 2)\cdot \left(\frac{m}{2} - (j \bmod 2)\right) \right) \\
|F^2_{i,j}| = &|F^2_{i,j,1}| + |F^2_{i,j,2}| =\\
= &n\left(
	(m \bmod 2)(1+m-1) + 
    (1 - m \bmod 2)(2(j \bmod 2) + m - 2(j \bmod 2))
\right) =\\
= &n\left(
	(m \bmod 2) \cdot m + 
    (1 - m \bmod 2) \cdot m
\right) = nm
\end{align*}

Similar to the matchings in $\mathfrak{F}^1$, the matchings $F^2_{i,j}$ and $F^2_{i',j'}$ are disjoint if $i \ne i'$ or $j \ne j'$. Note that for every $i=2,3,\ldots,2n-1-\sigma_n$, $F^2_{i,j}$ is $jn$-splitted perfect matching for $\mathbf{v}$ for every $j=1,2,\ldots,m-1$, and is a non-splitted perfect matching if $j=0$.

The set $E^3$ contains the preimages of the edges of the non-splitted perfect matching $N^0_1$ of $K_{2n}$ and the preimages of all edges of $K_{2m}[\{\widetilde{u}^1,\widetilde{u}^2,\ldots,\widetilde{u}^m\}] \cup K_{2m}[\{\widetilde{v}^1,\widetilde{v}^2,\ldots,\widetilde{v}^m\}]$. The preimages of the edges of $K_{2m}$ form $2n$ disjoint complete graphs on $m$ vertices, namely $K_{2mn}\left[\{x_s^1,x_s^2,\ldots,x_s^m\}\right]$, for every $\overline{x}_s \in V(K_{2n})$. For every edge $\overline{x}_s\overline{y}_t \in N^0_1$, its preimages together with the two copies of $K_{m}$ corresponding to the vertices $\overline{x}_s$ and $\overline{y}_t$ form the subgraph $K_{2mn}\left[\{x_s^1,y_t^1,x_s^2,y_t^2,\ldots,x_s^m,y_t^m\}\right]$, which is isomorphic to $K_{2m}$. So, the set $E^3$ consists of $n$ disjoint copies of $K_{2m}$. For every perfect matching $M \in \widetilde{\mathfrak{F}}$ we construct one perfect matching in $K_{2mn}$ by joining its $n$ disjoint copies in $E^3$: 
\begin{align*}
F^3_{i} = &
\bigcup\limits_{\substack{\overline{x}_s\overline{y}_t \in N^0_1}}
\left\{ 
\{x_s^{p}x_s^{q}\ |\ \widetilde{u}^{p}\widetilde{u}^{q} \in M_i\} 
\cup 
\{x_s^{p}y_t^{q}\ |\ \widetilde{u}^{p}\widetilde{v}^{q} \in M_i\} 
\cup
\{y_t^{p}y_t^{q}\ |\ \widetilde{v}^{p}\widetilde{v}^{q} \in M_i\} 
\right\}\\
F'^3_{i} = &
\bigcup\limits_{\substack{\overline{x}_s\overline{y}_t \in N^0_1}}
\left\{ 
\{x_s^{p}x_s^{q}\ |\ \widetilde{u}^{p}\widetilde{u}^{q} \in M^0_i\} 
\cup 
\{x_s^{p}y_t^{q}\ |\ \widetilde{u}^{p}\widetilde{v}^{q} \in M^0_i\} 
\cup
\{y_t^{p}y_t^{q}\ |\ \widetilde{v}^{p}\widetilde{v}^{q} \in M^0_i\} 
\right\}\\
\mathfrak{F}^3 = &\left\{F^3_{i}\ |\ i=1,2,\ldots,\sigma_m\right\} \cup \left\{F'^3_{i}\ |\ i=1,2,\ldots,2m-1-\sigma_m\right\}
\end{align*}
The sets $F^3_i$ and $F'^3_i$ are pairwise disjoint matchings having $mn$ edges each. Note that if $M_i$ is $r$-splitted perfect matching for $\widetilde{\mathbf{v}}$, then $F^3_i$ is $rn$-splitted perfect matching for $\mathbf{v}$, $i=1,2,\ldots,\sigma_m$. Moreover, the perfect matchings $F'^3_i$ are not splitted, $i=1,2,\ldots,2m-1-\sigma_m$.

The number of the constructed perfect matchings in $\mathfrak{F}$ is $m\sigma_n + m(2n-2-\sigma_n) + 2m-1 = 2mn-1$. Out of these the number of splitted perfect matchings is $m\sigma_n + (m-1)(2n-2-\sigma_n) + \sigma_m = \sigma_m+\sigma_n+2(m-1)(n-1)$. This completes the proof.
\end{proof}

By applying Equivalence theorem we obtain the following lower bound, which is a generalization of Theorem \ref{tPetrosyan4n}:
\begin{theorem}
\label{tComposite}
For any $m,n \in\mathbb{N}$, $W(K_{2mn}) \geq W(K_{2m}) + W(K_{2n}) + 4(m-1)(n-1) - 1$.
\end{theorem}

We know that $W(K_6) = 7$ and $W(K_{10}) \geq 14$. The above theorem implies that $W(K_{30}) \geq 52$. This result disproves Conjecture \ref{conjLog} which predicted that $W(K_{30})=51$. But this is not the smallest case that contradicts the conjecture as we will see in Section \ref{sExact}.

\begin{corollary}
\label{cLower}
If $n=\prod\limits_{i=1}^{\infty}p_i^{\alpha_i}$, where $p_i$ is the $i$-th prime number, $\alpha_i \in \mathbb{Z}_+$, then
\begin{center}
$W(K_{2n}) \geq 4n - 3 - \sum\limits_{i=1}^{\infty}{\alpha_i\left(4p_i-3-W(K_{2p_i})\right)}$.
\end{center}
\end{corollary}
\begin{proof}
Let $d_m$ denote the difference $W(K_{2m}) - (4m - 3)$. Theorem \ref{tComposite} states that $d_{mk} \geq d_m + d_k$. By induction we get $d_n \geq \sum_{i=1}^{\infty}{\alpha_i d_{p_i}}$. We complete the proof by replacing $d_{p_i}$ by its value.
\end{proof}

\section{Upper bounds}
Let $\alpha$ be an arbitrary interval edge-coloring of $K_{2n}$, $n \in \mathbb{N}$, and $\mathbf{v}_\alpha = \left(u_1,v_1, u_2,v_2, \ldots,u_n,v_n\right)$ be its corresponding ordering of vertices. Let the shift vector of $\alpha$ be ${\rm sh}(\alpha) = (b_1,b_2,\ldots,b_{n-1})$. Equivalence lemma implies that there exists a 1-factorization of $K_{2n}$ $\mathfrak{F} = 
\left\{ F^0_j\ |\ j=1,2,\ldots,2n-1-\sum\limits_{i=1}^{n-1}b_i \right\}
\cup
\bigcup\limits_{i=1}^{n-1}\left\{F^i_j\ |\ j=1,2,\ldots,b_i\right\}$, such that $F^i_j$ is $i$-splitted with respect to the ordering $\mathbf{v_\alpha}$, $i=1,2,\ldots,n-1$, $j=1,2,\ldots,b_i$. Wherever we have an interval coloring $\alpha$ of a complete graph in the proofs of this section we will always assume that the corresponding ordering of vertices $\mathbf{v}_{\alpha}$ and 1-factorization $\mathfrak{F}$ is given.

To improve the upper bounds on $W(K_{2n})$ we need several lemmas.
\begin{lemma}
\label{lReverse}
If for some interval edge-coloring $\alpha$ of $K_{2n}$, ${\rm sh}(\alpha) = (b_1,b_2,\ldots, b_{n-1})$, then there exists interval edge-coloring $\beta$ of $K_{2n}$ such that ${\rm sh}(\beta) = (b_{n-1}, b_{n-2}, \ldots, b_1)$.
\end{lemma}
\begin{proof}
Note that if some $F \in \mathfrak{F}$ is $i$-splitted with respect to $\mathbf{v}_\alpha$, then it is $(n-i)$-splitted with respect to the ordering $\mathbf{v}'_\alpha = \left(u_n,v_n,u_{n-1},v_{n-1},\ldots,u_1,v_1\right)$. We use Equivalence lemma to construct a coloring $\beta$ from $\mathfrak{F}$ with respect to the ordering $\mathbf{v}'_\alpha$. Its shift vector is $(b_{n-1}, b_{n-2}, \ldots, b_1)$.
\end{proof}

\begin{lemma}
\label{lLessColors}
If for some interval edge-coloring $\alpha$ of $K_{2n}$, ${\rm sh}(\alpha) = (b_1,b_2,\ldots, b_{n-1})$, where $b_i > 0$ for some $i \in [1,n-1]$, then there exists interval edge-coloring $\beta$ of $K_{2n}$ such that ${\rm sh}(\beta) = (b_1, b_2, \ldots, b_{i-1}, b_{i}-1, b_{i+1}, \ldots, b_{n-1})$.
\end{lemma}
\begin{proof}
The condition $b_i > 0$ implies that there exists a perfect matching $F_{b_i}^i \in \mathfrak{F}$ which is $i$-splitted with respect to the ordering $\mathbf{v}_\alpha$. We construct the coloring $\beta$ by applying Equivalence lemma to the 1-factorization $\mathfrak{F}$ by regarding the perfect matching $F_{b_i}^i$ as a non-splitted one (we can rename it to $F^0_{2n-|{\rm sh}(\alpha)|}$).
\end{proof}

\begin{lemma}
\label{l2k1}
If ${\rm sh}(\alpha) = (b_1,b_2,\ldots, b_{n-1})$ for some interval edge-coloring $\alpha$ of $K_{2n}$, then  
\begin{center}
$\sum\limits_{i=1}^{k}{b_i} \leq 2k-1$, for every $k=1,2,\ldots,n-1$.
\end{center}
\end{lemma}

\begin{proof}
According to the proof of Equivalence lemma, the left parts of the perfect matchings $F^i_j$ cover the vertex $u_1$ (and $v_1$), $i=1,2,\ldots,k$, $j=1,2,\ldots,b_i$. Moreover, 
\begin{center}
$\bigcup\limits_{i=1}^{k}
\bigcup\limits_{j=1}^{b_i}
{l^i_{\mathbf{v}_\alpha}\left(F^i_j\right)} 
\subset E\left(H_{\mathbf{v}_\alpha}^{[1,k]}\right)$
\end{center}
To complete the proof we note that the number of the perfect matchings $F^i_j$ is $\sum\limits_{i=1}^{k}{b_i}$ and the degree of the vertex $u_1$ (or $v_1$) in $H_{\mathbf{v}_\alpha}^{[1,k]}$ is $2k-1$.
\end{proof}

We will call the vector $(b_1,b_2,\ldots,b_k)$ \textit{saturated} if $\sum\limits_{i=1}^{k}{b_i} = 2k-1$.

\begin{corollary}
\label{c2n4}
If $\alpha$ is an interval edge-coloring of $K_{2n}$, $n\geq 3$, then  
$|{\rm sh}(\alpha)| \leq 2n-4$.
\end{corollary}

\begin{proof}
Let ${\rm sh}(\alpha) = (b_1,b_2,\ldots, b_{n-1})$. Lemma \ref{l2k1} implies that $\sum\limits_{i=1}^{n-2}{b_i} \leq 2n-5$. The same lemma in conjuction with Lemma \ref{lReverse} implies that $b_{n-1} \leq 1$. By summing these two inequalties we complete the proof.
\end{proof}

\begin{lemma}
\label{lAfterSaturated}
If ${\rm sh}(\alpha) = (b_1,b_2,\ldots, b_{n-1})$ for some interval edge-coloring $\alpha$ of $K_{2n}$ and $(b_1,b_2,\ldots,b_k)$ is saturated for some $k \in [2,n-2]$, then $b_{k+1} \leq 1$.
\end{lemma}

\begin{proof}
Lemma \ref{l2k1} implies that $b_{k+1} \leq 2$. To complete the proof we need to show that $b_{k+1} \neq 2$. Suppose the contrary, $b_{k+1} = 2$. $(b_1,b_2,\ldots,b_k)$ is saturated, so the proof of Lemma \ref{l2k1} implies that the edges $u_1x_i$ and $v_1x_i$, $x \in \{u,v\}$, $i=2,3,\ldots,k$, belong to the perfect matchings $F^i_j$, $i=1,2,\ldots,k$, $j=1,2,\ldots,b_i$. Similarly, the edges $u_1u_{k+1}$, $u_1v_{k+1}$, $v_1u_{k+1}$ and $v_1v_{k+1}$ must be covered by $F^{k+1}_1$ and $F^{k+1}_2$.

Now we look at the vertex $u_2$. It is covered by the left parts of the perfect matchings $F^i_j$, $i=2,3,\ldots,k$, $j=1,2,\ldots,b_i$. In total these matchings cover all but $2k-1 - \sum\limits_{i=2}^{k}{b_i} = b_1$ edges incident to $u_2$ in the subgraph $H_{\mathbf{v}_\alpha}^{[1,k]}$. Lemma \ref{l2k1} implies that $b_1 \leq 1$, so at most one edge is left uncovered. The vertex $u_2$ must be covered by the left parts of $F^{k+1}_1$ and $F^{k+1}_2$ as well. The edges $u_2u_{k+1}$ and $u_2v_{k+1}$ cannot be used as the vertices $u_{k+1}$ and $v_{k+1}$ are already covered by $F^{k+1}_1$ and $F^{k+1}_2$. Therefore, at most one edge remains for these two matchings, which is a contradiction.
\end{proof}

\begin{corollary}
\label{c2n5}
If $\alpha$ is an interval edge-coloring of $K_{2n}$, $n\geq 5$, then  
$|{\rm sh}(\alpha)| \leq 2n-5$.
\end{corollary}

\begin{proof}
Let ${\rm sh}(\alpha) = (b_1,b_2,\ldots, b_{n-1})$. Lemma \ref{l2k1} implies that $\sum\limits_{i=1}^{n-3}{b_i} \leq 2n-7$. We consider two cases.
\begin{description}
\item{Case 1:} $\sum\limits_{i=1}^{n-3}{b_i} = 2n-7$. Lemma \ref{lAfterSaturated} implies that $b_{n-2} \leq 1$. Lemmas \ref{lReverse} and \ref{l2k1} imply that $b_{n-1} \leq 1$. The sum of these inequalities proves the required bound.
\item{Case 2:} $\sum\limits_{i=1}^{n-3}{b_i} \leq 2n-8$. Lemmas \ref{lReverse} and \ref{l2k1} imply that $b_{n-2} + b_{n-1} \leq 3$. The sum of these inequalities completes the proof.
\end{description}
\end{proof}

\begin{lemma}
\label{lBeforeSaturated}
If ${\rm sh}(\alpha) = (b_1,b_2,\ldots, b_{n-1})$ for some interval edge-coloring $\alpha$ of $K_{2n}$ and $(b_1,b_2,\ldots,b_k)$ is saturated for some $k \in [3,n-1]$, then $b_{k} \geq 3$.
\end{lemma}

\begin{proof}
Suppose the contrary, $b_k \leq 2$. If $b_k=2$, then the vector $(b_1,b_2,\ldots,b_{k-1})$ is also saturated, and we obtain contradiction with Lemma \ref{lAfterSaturated}. If $b_k \leq 1$, then we have $\sum\limits_{i=1}^{k-1}{b_i} \geq 2k-2$ which contradicts Lemma \ref{l2k1}. 
\end{proof}

\begin{lemma}
\label{lEdgeCount}
If ${\rm sh}(\alpha) = (b_1,b_2,\ldots, b_{n-1})$ for some interval edge-coloring $\alpha$ of $K_{2n}$ and $k \in [2,n-2]$, then
\begin{center}
$k(2k-1) \geq \sum\limits_{i=1}^{k}{ib_i} + \sum\limits_{i=k+1}^{\min\{2k-1,n-1\}}{(2k-i)b_i}$.
\end{center}
\end{lemma}
\begin{proof}
We consider the subgraph $H_{\mathbf{v}_\alpha}^{[1,k]}$. The number of edges in the subgraph is $k(2k-1)$. The left part of each of the perfect matchings $F^i_j$, $i=1,2,\ldots,k$, $j=1,2,\ldots,b_i$, consists of $i$ edges, and all of them belong to the subgraph $H_{\mathbf{v}_\alpha}^{[1,k]}$. The number of such edges is $\sum\limits_{i=1}^{k}{ib_i}$. 

Now we fix an $i \in [k+1,r]$, where $r$ denotes $\min\{2k-1,n-1\}$. The left part of each of the perfect matchings $F_j^i$, $j=1,2,\ldots,b_i$, consists of $i$ edges. At most $2i-2k$ of them can join some vertex from $H_{\mathbf{v}_\alpha}^{[1,k]}$ with some vertex from $H_{\mathbf{v}_\alpha}^{[k+1,i]}$. So at least $2k-i$ edges belong to the subgraph $H_{\mathbf{v}_\alpha}^{[1,k]}$. The number of such edges is at least $\sum\limits_{i=k+1}^{r}{(2k-i)b_i}$.
\end{proof}

Lemma \ref{lEdgeCount} implies that if for some fixed $k_0$ there are many $i$-splitted perfect matchings where $i \leq k_0$, then there cannot be too many $i'$-splitted perfect matchings where $k_0 < i' < \min\{2k_0-1,n-1\}$. In order to use this lemma we need to bound the sum $\sum\limits_{i=1}^{k}{ib_i}$ from below.

For the numbers $k\in \mathbb{N}$ and $r \in \mathbb{Z}_+$ we define the following:
\begin{align*}
&T_k = \left\{\left(b_1,b_2,\ldots,b_k\right)\ |\ \exists \alpha \text{ interval coloring of }K_{2n},\ n>k,\ {\rm sh}(\alpha)=(b_1,b_2,\ldots,b_{n-1})\right\} \\
&m(k,r) = \min\limits_{(b_1,b_2,\ldots,b_{k}) \in T_k}\left\{\sum\limits_{i=1}^{k}{ib_i}\ |\ \sum\limits_{i=1}^{k}{b_i}=r\right\}
\end{align*}
Note that $m(k,r)$ is not defined for all pairs $(k,r)$. For example, Lemma \ref{l2k1} implies that there are no interval colorings of $K_{2n}$ for which $\sum\limits_{i=1}^{k}{b_i}=r$ if $r \geq 2k$. It is obvious that $m(1,1)=1$ and $m(k,0)=0$, $k \in \mathbb{N}$.

\begin{remark}
In order to calculate $m(k,r)$, $k>1$, $r>1$, it is sufficient to take the minimum over those $(b_1,b_2,\ldots,b_k) \in T_k$ for which $\sum\limits_{i=1}^{k-1}{ib_i} = m(k-1,r-b_k)$.
\end{remark}

Table \ref{tableMkr} lists the values of $m(k,r)$ for $k\leq 4$ and $r \leq 7$. For example, $m(3,5)$ is calculated as follows. According to the above remark the possible candidate vectors from $T_3$ are $(1,2,2)$, $(1,1,3)$, $(1,0,4)$ and $(0,0,5)$. Lemma \ref{lAfterSaturated} implies that $(1,2,2) \notin T_3$. The coloring of $K_{12}$ in Fig. \ref{fK12} proves that $(1,1,3) \in T_3$. On the other hand, sum $b_1 + 2b_2 + 3b_3$ is larger for the other two candidate vectors, so $m(3,5) = 12$. Similarly we show that $m(4,7)=20$ and the minimum is achieved on the vector $(1,2,1,3)$, which clearly belongs to $T_4$ as illustrated in the coloring of $K_{22}$ in Fig. \ref{fK22}. By applying Lemma \ref{lLessColors} to these two colorings we prove that all the other vectors listed in the Table \ref{tableMkr} belong to the corresponding $T_k$'s.

\begin{figure}
\centering
\includegraphics[width=0.7\textwidth]{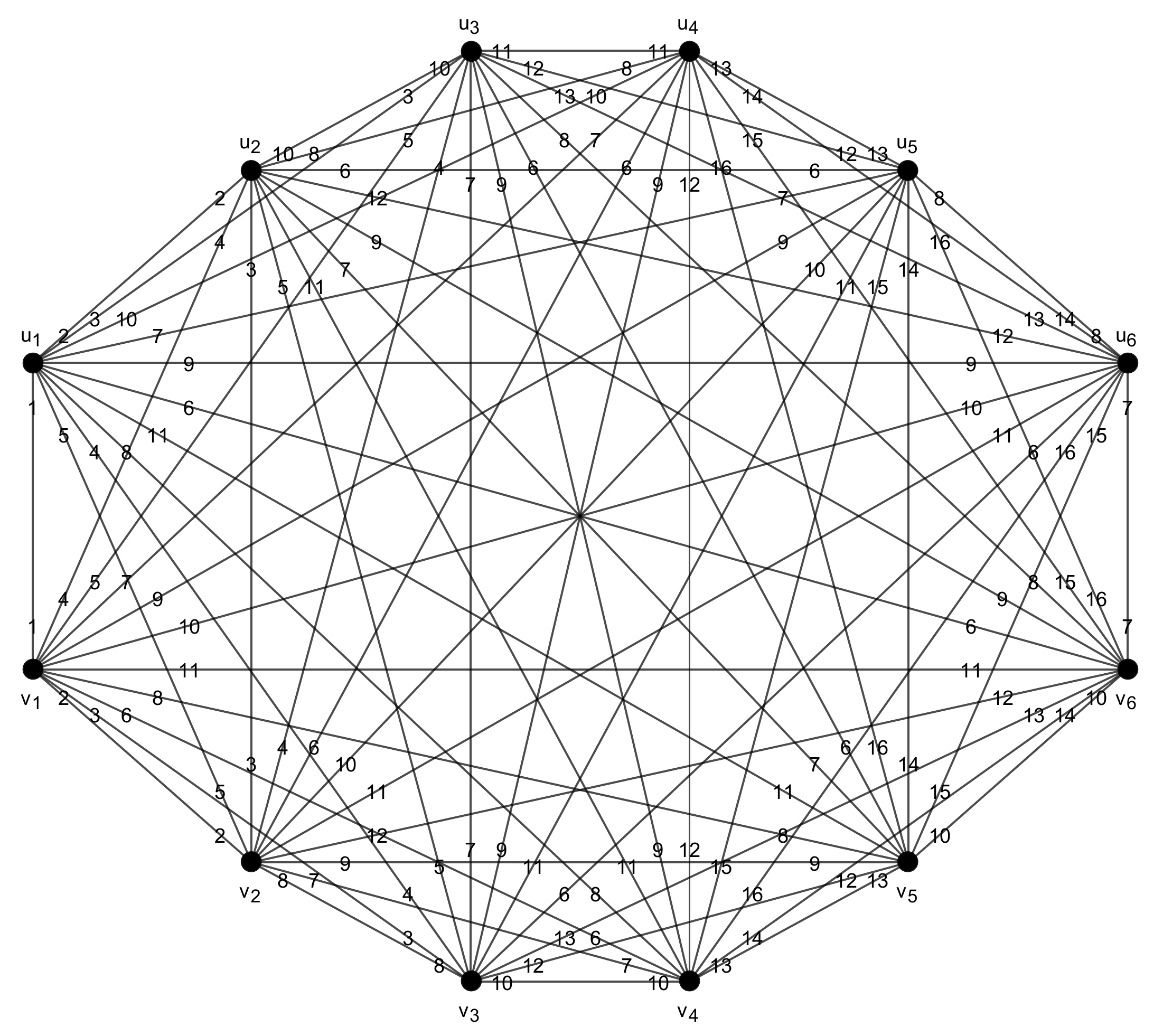}
\caption{Interval $16$-coloring of $K_{12}$ with a shift vector $(1,1,3,0,0)$.}
\label{fK12}
\end{figure}

\begin{table}[t!]
\centering
\newcolumntype{C}{>{\centering\arraybackslash}X}%

\begin{tabularx}{0.8\textwidth}{c||C|C|C|C|}
\backslashbox{$r$}{$k$} & $1$ & $2$ & $3$ & $4$ \\ \hline\hline
$0$
& \begin{tabular}{c}$0$ \\ $(0)$ \end{tabular} 
& \begin{tabular}{c}$0$ \\ $(0,0)$ \end{tabular} 
& \begin{tabular}{c}$0$ \\ $(0,0,0)$ \end{tabular}  
& \begin{tabular}{c}$0$ \\ $(0,0,0,0)$ \end{tabular} \\ \hline

$1$ 
& \begin{tabular}{c}$1$ \\ $(1)$ \end{tabular} 
& \begin{tabular}{c}$1$ \\ $(1,0)$ \end{tabular} 
& \begin{tabular}{c}$1$ \\ $(1,0,0)$ \end{tabular} 
& \begin{tabular}{c}$1$ \\ $(1,0,0,0)$ \end{tabular} \\ \hline

$2$
&
& \begin{tabular}{c}$3$ \\ $(1,1)$ \end{tabular}  
& \begin{tabular}{c}$3$ \\ $(1,1,0)$ \end{tabular}  
& \begin{tabular}{c}$3$ \\ $(1,1,0,0)$ \end{tabular} \\ \hline

$3$
&
& \begin{tabular}{c}$5$ \\ $(1,2)$ \end{tabular}  
& \begin{tabular}{c}$5$ \\ $(1,2,0)$ \end{tabular} 
& \begin{tabular}{c}$5$ \\ $(1,2,0,0)$ \end{tabular} \\ \hline

$4$
&
&
& \begin{tabular}{c}$8$ \\ $(1,2,1)$ \end{tabular} 
& \begin{tabular}{c}$8$ \\ $(1,2,1,0)$ \end{tabular} \\ \hline

$5$
&
&
& \begin{tabular}{c}$12$ \\ $(1,1,3)$ \end{tabular} 
& \begin{tabular}{c}$12$ \\ $(1,2,1,1)$ \end{tabular} \\ \hline

$6$
&
&
&  
& \begin{tabular}{c}$16$ \\ $(1,2,1,2)$ \end{tabular} \\ \hline

$7$
&
&
&  
& \begin{tabular}{c}$20$ \\ $(1,2,1,3)$ \end{tabular} \\ \hline
\end{tabularx}

\caption{
	The values of $m(k,r)$. The first row of each of the cells displays the value of $m(k,r)$. The second row contains some vector $(b_1,b_2,\ldots,b_k) \in T_k$ for which $\sum\limits_{i=1}^{k}{ib_i} = m(k,r)$.
}
\label{tableMkr}
\end{table}

\begin{lemma}
\label{l2n6}
If $\alpha$ is an interval edge-coloring of $K_{2n}$, $n\geq 9$, then  
$|{\rm sh}(\alpha)| \leq 2n-6$.
\end{lemma}
\begin{proof}
Suppose the contrary, $|{\rm sh}(\alpha)| \geq 2n-5$. Lemmas \ref{lReverse} and \ref{l2k1} imply that $\sum\limits_{i=5}^{n-1}{b_i} \leq 2n-11$. We consider three cases.
\begin{description}
\item{Case 1:} $\sum\limits_{i=5}^{n-1}{b_i} = 2n-11$. Lemmas \ref{lReverse}, \ref{lBeforeSaturated} and \ref{lAfterSaturated} imply that $b_{5} \geq 3$ and $b_{4} \leq 1$. We apply Lemma \ref{l2k1} for $k=3$ to show that $b_{1} + b_{2} + b_{3} = 5$ and $b_{4}=1$. Then we apply Lemma \ref{lEdgeCount} for $k=3$. The left part of the inequality is $15$. On the right side we have $\sum\limits_{i=1}^{3}{ib_i} \geq m(3,5)=12$ and $\sum\limits_{i=4}^{5}{(6-i)b_i} \geq 5$. These inequalities contradict Lemma \ref{lEdgeCount}.

\item{Case 2:} $\sum\limits_{i=5}^{n-1}{b_i} = 2n-12$. Lemma \ref{l2k1} implies that $(b_1,b_2,b_3,b_4)$ is saturated. Lemma \ref{lAfterSaturated} implies that $b_5 \leq 1$. Therefore, $(b_{n-1},b_{n-2},\ldots,b_6)$ is saturated and $b_5=1$. Lemma \ref{lBeforeSaturated} implies that $b_6 \geq 3$. Now we apply Lemma \ref{lEdgeCount} for $k=4$. The left part of the inequality is $28$. On the right side, $\sum\limits_{i=1}^{4}{ib_i} \geq m(4,7)=20$ and $\sum\limits_{i=5}^{7}{(8-i)b_i} \geq 9$. These inequalities contradict Lemma \ref{lEdgeCount}.

\item{Case 3:} $\sum\limits_{i=5}^{n-1}{b_i} \leq 2n-13$. Lemma \ref{l2k1} implies that $\sum\limits_{i=1}^{4}{b_i} \leq 7$. By summing these two inequalities we obtain a contradiction.
\end{description}
\end{proof}

\begin{figure}
\centering
\includegraphics[width=\textwidth]{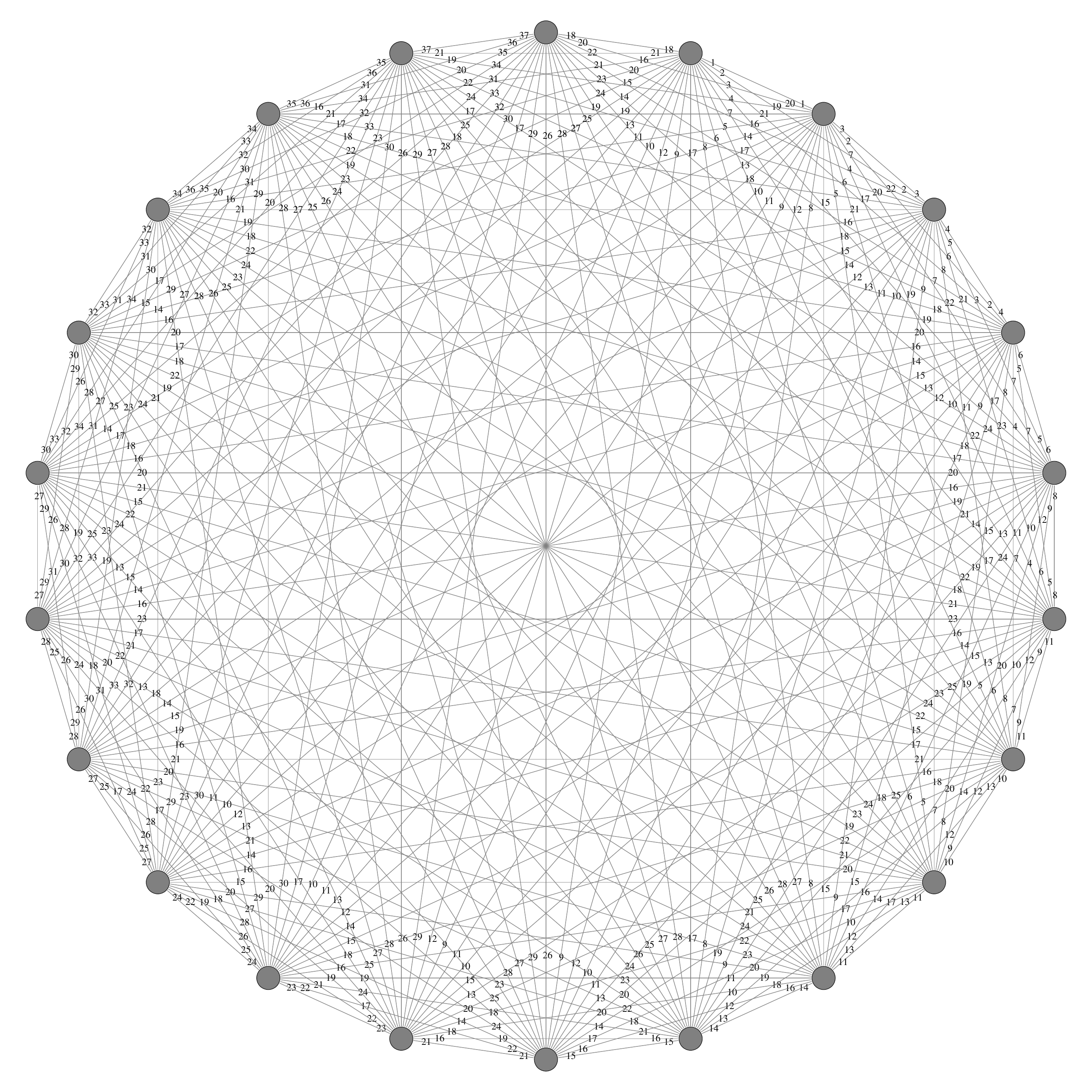}
\caption{Interval $37$-coloring of $K_{22}$ with a shift vector $(1,2,1,3,1,1,3,1,2,1)$.}
\label{fK22}
\end{figure}

Corollaries \ref{c2n4}, \ref{c2n5}, Lemma \ref{l2n6} and Remark \ref{totalShift} imply the following upper bound on $W(K_{2n})$.

\begin{theorem}
\label{tUpper}
If $n \geq 3$, then
\begin{center}
$W(K_{2n}) \leq \left\{
\begin{tabular}{ll}
$4n-5$, & if $n \geq 3$,\\
$4n-6$, & if $n \geq 5$,\\
$4n-7$, & if $n \geq 9$.\\
\end{tabular}
\right.$
\end{center}
\end{theorem}

\section{More exact values and an improved lower bound}
\label{sExact}

The lower bound on $W(K_{2n})$ from Corollary \ref{cLower} depends on the values $W(K_{2p})$ where $p$ is a prime number. For $p=2$ and $p=3$ the exact values of $W(K_{2p})$ were known before \cite{Petrosyan2010}. For $p=5$ the lower bound from Theorem \ref{t35n3} coincides with the upper bound from Theorem \ref{tUpper}. The case $p=7$ is resolved by the lemma below. Finally, for the case $p=11$, the upper bound from Theorem \ref{tUpper} is achieved by the interval $37$-coloring of $K_{22}$ shown in Fig. \ref{fK22}. This coloring also rejects Conjecture \ref{conjLog}, which predicts that $W(K_{22})=36$.

\begin{lemma}
\label{lK14}
$W(K_{14}) = 21$.
\end{lemma}
\begin{proof}
Theorem \ref{tUpper} implies that $W(K_{14}) \leq 22$. It is sufficient to show that $K_{14}$ does not have an interval coloring with $22$ colors. Suppose the contrary, there exists $\alpha$ interval $22$-coloring of $K_{14}$.

Consider its shift vector ${\rm sh}(\alpha)=(b_1,b_2,b_3,b_4,b_5,b_6)$. From Remark \ref{totalShift} we have that $\sum\limits_{i=1}^{6}{b_i}=9$. Lemma \ref{l2k1} implies that the sums of both first and last triples cannot exceed $5$. Without loss of generality we can assume that $b_1+b_2+b_3=5$ and $b_4+b_5+b_6=4$. Lemma \ref{lAfterSaturated} implies that $b_4 \leq 1$. Lemmas \ref{lReverse} and \ref{l2k1} imply that $b_5+b_6=3$ and $b_4=1$. So $b_5 \geq 2$. 

Now we check the inequality from Lemma \ref{lEdgeCount} for $k=3$. The left part equals $15$. On the right part we have $\sum\limits_{i=1}^{3}{ib_i} \geq m(3,5) = 12$, $\sum\limits_{i=4}^{5}{(6-i)b_i} \geq 4$. By summing these two inequalties we get a contradiction.
\end{proof}

The best lower bound we could obtain is the following. 

\begin{theorem}
\label{tLower}
If $n = \prod\limits_{i=1}^{\infty}{p_i^{\alpha_i}}$, where $p_i$ is the $i$-th prime number and $\alpha_i \in \mathbb{Z}_+$, then
\begin{center}
$W(K_{2n}) \geq 4n - 3 - \alpha_1 - 2\alpha_2 - 3\alpha_3 - 4\alpha_4 - 4\alpha_5 - \frac{1}{2}\sum\limits_{i=6}^{\infty}{\alpha_i(p_i+1)} $.
\end{center}
\end{theorem}
\begin{proof}
To prove the bound we take the bound from Corollary \ref{cLower}, set the exact values of $W(K_{2p_i})$ for the first five prime numbers and use Theorem \ref{t35n3} to bound $W(K_{2p_i})$ for $i\geq 6$, taking into account that all prime numbers except $2$ are odd.
\end{proof}

Table \ref{tableAll} lists obtained lower and upper bounds on $W(K_{2n})$ and all known exact values for $n \leq 18$.

\begin{table}[h]
\centering
\begin{tabularx}{0.96\textwidth}{r||*{18}{c|}}
$n$ 
& $1$ & $2$ & $3$ & $4$ & $5$ & $6$ & $7$ & $8$ & $9$ & $10$ & $11$ & $12$ & $13$ & $14$ & $15$ & $16$ & $17$ & $18$ \\ \hline\hline
$W(K_{2n}) \geq $ 
& $1$ & $4$ & $7$ & $11$ & $14$ & $18$ & $21$ & $26$ & $29$ & $33$ & $37$ & $41$ & $42$ & $46$ & $52$ & $57$ & $56$ & $64$ \\ \hline
$W(K_{2n}) = $
& $1$ & $4$ & $7$ & $11$ & $14$ & $18$ & $21$ & $26$ & $29$ & $33$ & $37$ & $41$ &  &  &  & $57$ &  &    \\ \hline
$W(K_{2n}) \leq $
& $1$ & $4$ & $7$ & $11$ & $14$ & $18$ & $22$ & $26$ & $29$ & $33$ & $37$ & $41$ & $45$ & $49$ & $53$ & $57$ & $61$ & $65$
\end{tabularx}
\caption{
	Bounds on $W(K_{2n})$: The first row lists the lower bounds from Theorem \ref{tLower}, the second row lists the known exact values and the third row lists the upper bounds from Theorem \ref{tUpper}. 
}
\label{tableAll}
\end{table}

\section*{Acknowledgements}
We would like to thank the organizers of the 7th Cracow Conference on Graph Theory for the wonderful atmosphere. We also thank Attila Kiss for suggesting the term \textit{shift vector}. We also would like to thank the reviewers for many valuable comments. This work was made possible by a research grant from the Armenian National Science and Education Fund (ANSEF) based in New York, USA.


\begin{thebibliography}{99}
\bibitem{AsratianKamalian1987} A.S. Asratian, R.R. Kamalian, \textit{Interval
colorings of edges of a multigraph}, Appl. Math. \textbf{5} (1987), 25--34 (in Russian).
\bibitem{AsratianKamalian1994} A.S. Asratian, R.R. Kamalian, \textit{Investigation on
interval edge-colorings of graphs}, J. Combin. Theory Ser. B \textbf{62} (1994) 34--43. doi:10.1006/jctb.1994.1053
\bibitem{Axenovich2002} M.A. Axenovich, \textit{On interval colorings of planar graphs},
Congr. Numer. \textbf{159} (2002), 77--94.
\bibitem{GiaroKubMal2001} K. Giaro, M. Kubale, M. Malafiejski, \textit{Consecutive
colorings of the edges of general graphs}, Discrete Math. \textbf{236} (2001), 131--143. doi:10.1016/S0012-365X(00)00437-4
\bibitem{Kamalian1989} R.R. Kamalian, \textit{Interval colorings of complete bipartite
graphs and trees}, preprint, Comp. Cen. of Acad. Sci. of Armenian SSR, Erevan,
1989 (in Russian).
\bibitem{Kamalian1990} R.R. Kamalian, \textit{Interval edge colorings of graphs},
Doctoral Thesis, Novosibirsk, 1990.
\bibitem{KamalianPetrosyan2012} R.R. Kamalian, P.A. Petrosyan, \textit{A note on
interval edge-colorings of graphs}, Mathematical problems of computer science \textbf{36} (2012), 13--16.
\bibitem{Khachatrian2012} H. Khachatrian, \textit{Investigation on interval
edge-colorings of Cartesian products of graphs}, Yerevan State University, BS
thesis, 2012 (in Armenian).
\bibitem{Konig1916} D. König, \textit{Uber Graphen und ihre Anwendung auf Determinantentheorie und Mengenlehre}, Math. Ann. \textbf{77} (1916), 453--465.
\bibitem{Petrosyan2005} P.A. Petrosyan, \textit{Interval edge-colorings of Möbius ladders}, Proceedings of the CSIT Conference (2005), 146--149 (in Russian).
\bibitem{Petrosyan2010} P.A. Petrosyan, \textit{Interval edge-colorings of complete
graphs and $n$-dimensional cubes}, Discrete Math. \textbf{310} (2010), 1580--1587.
 doi:10.1016/j.disc.2010.02.001
\bibitem{PetrosyanKhachatrianTananyan2013} P.A. Petrosyan, H.H. Khachatrian,
H.G. Tananyan, \textit{Interval edge-colorings of Cartesian products of graphs I},
Discuss. Math. Graph Theory \textbf{33} (2013), 613--632. doi:10.7151/dmgt.1693
\bibitem{Vizing1965} V.G. Vizing, \textit{The chromatic class of a multigraph},
Kibernetika \textbf{3} (1965), 29--39 (in Russian)
\bibitem{West} D.B. West, \textit{Introduction to Graph Theory}, Prentice-Hall, New Jersey, 1996.
\end{thebibliography}
\end{document}